\documentclass[sigplan,nonacm,screen]{acmart}
\AtBeginDocument{%
  }

\usepackage[ruled, vlined, linesnumbered]{algorithm2e}
\usepackage{graphicx}
\usepackage{textcomp}
\usepackage[table]{xcolor}
\usepackage{booktabs}
\usepackage{hyperref}
\usepackage{makecell}
\usepackage{adjustbox}
\usepackage{cleveref}
\usepackage{multirow}
\usepackage{soul}
\usepackage{lipsum}
\usepackage{pifont}
\usepackage{balance}
\usepackage{cuted}
\usepackage{amsmath,amsfonts}
\usepackage[most]{tcolorbox} 
\usepackage{colortbl}      
\usepackage{subcaption}
\usepackage{tikz}
\usetikzlibrary{quantikz2,shapes.geometric, arrows, positioning, fit, calc}
\usepackage{stfloats}

\definecolor{lightgreen}{RGB}{230, 255, 230}

\usetikzlibrary{positioning}       

\tikzset{
    block/.style={draw, rectangle, minimum height=2em, minimum width=4em, text width=3cm, align=center},
    line/.style={draw, -latex}    
}

\newtheorem{thm}{Theorem}\crefname{thm}{Theorem}{Theorems}
\newtheorem{lem}[thm]{Lemma}\crefname{lem}{Lemma}{Lemmas}
\crefname{prp}{Proposition}{Propositions}
\newtheorem{cor}[thm]{Corollary}\crefname{cor}{Corollary}{Corollaries}
\crefname{prb}{Problem}{Problems}
\crefname{dfn}{Definition}{Definitions}
\crefname{conj}{Conjecture}{Conjectures}

\newcommand{\tarc}{\mbox{\large$\frown$}}
\newcommand{\arc}[2][-3ex]{{#2}{\kern #1{\raisebox{1.5ex}{\tarc}}}}

\definecolor{CustomOrange}{RGB}{200, 120, 50} 
\definecolor{DarkTeal}{RGB}{0, 105, 120}     
\definecolor{DarkBlue}{RGB}{0, 75, 135}      
\definecolor{DarkOlive}{RGB}{85, 107, 47}    
\definecolor{DarkMaroon}{RGB}{94, 38, 45}    
\definecolor{SoftOrange}{RGB}{255, 90, 70}  

\definecolor{WarmOrange}{RGB}{230, 130, 60}  
\definecolor{SkyBlue}{RGB}{100, 180, 220}

\definecolor{TextHighlightBlue}{RGB}{72,120,208}

\newcommand\note[1]{#1}

\newcommand{\dquote}[1]{``#1''}

\newcommand{\SQiSW}{\mathrm{SQ\lowercase{i}SW}}
\newcommand{\iSWAP}{\mathrm{\lowercase{i}SWAP}}
\newcommand{\CNOT}{\mathrm{CNOT}}
\newcommand{\CZ}{\mathrm{CZ}}

\newcommand{\ZZ}{\mathrm{ZZ}}

\newcommand{\SWAP}{\mathrm{SWAP}}

\newcommand{\B}{\mathrm{B}}
\newcommand{\SU}{\mathrm{SU}}
\newcommand{\Can}{\mathrm{Can}}
\newcommand{\fSim}{\mathrm{fSim}}
\newcommand{\CanGate}[3]{\mathrm{Can}(#1,#2,#3)}
\newcommand{\UGate}[3]{\mathrm{U_3}(#1,#2,#3)}
\newcommand{\EAplus}{{EA$_+$}}
\newcommand{\EAminus}{{EA$_-$}}

\newcommand{\CompilerEff}{ReQISC-Eff}
\newcommand{\CompilerFull}{ReQISC-Full}

\makeatletter
\newcommand{\removelatexerror}{\let\@latex@error\@gobble}
\makeatother

\newcommand{\code}{\texttt}
\newcommand{\numTwoQubit}{\#2Q}
\newcommand{\depthTwoQubit}{Depth2Q}

\SetAlFnt{\small}

\definecolor{ElegantGrayBack}{RGB}{242, 242, 242} 
\definecolor{ElegantGrayFrame}{RGB}{105, 105, 105} 

\newtcolorbox{takeaways}[1][]{
    colback=ElegantGrayBack, 
    colframe=ElegantGrayFrame, 
    fonttitle={\bfseries},
    title={Takeaways},       
    arc=1.5mm,                 
    boxrule=0.8pt,           
    left=5pt, right=5pt,     
    top=3pt, bottom=3pt,     
    enhanced,                
    breakable,               
    #1                       
}

\begin{document}


\title{Reconfigurable Quantum Instruction Set Computers for High Performance Attainable on Hardware}


\author{Zhaohui Yang}
\orcid{0000-0003-4698-4378}
\affiliation{%
  \institution{The Hong Kong University of\\ Science and Technology}
  \city{Hong Kong}
  \country{}
}
\email{zhaohui@ucsb.edu}

\author{Dawei Ding}
\orcid{0000-0001-7728-5380}
\affiliation{
  \institution{Fudan University}
  \institution{Shanghai Institute for Mathematics and Interdisciplinary Sciences}
  \city{Shanghai}
  \country{China}
  \country{}
}
\email{daweiding@fudan.edu.cn}

\author{Qi Ye}
\orcid{0009-0002-5606-8824}
\affiliation{
  \institution{Tsinghua University}
  \city{Beijing}
  \country{China}
}
\email{yeq22@mails.tsinghua.edu.cn}

\author{Cupjin Huang}
\orcid{0000-0002-7466-8033}
\affiliation{
  \institution{DAMO Academy, Alibaba Group}
  \city{Bellevue, WA}
  \country{USA}
}
\email{pertox4726@gmail.com}

\author{Jianxin Chen}
\orcid{0000-0002-9365-776X}
\authornote{Corresponding author. Email: \href{mailto:chenjianxin@tsinghua.edu.cn}{chenjianxin@tsinghua.edu.cn}.}
\affiliation{
  \institution{Tsinghua University}
  \city{Beijing}
  \country{China}
}
\email{chenjianxin@tsinghua.edu.cn}

\author{Yuan Xie}
\orcid{0000-0003-2093-1788}
\affiliation{
  \institution{The Hong Kong University of\\ Science and Technology}
  \city{Hong Kong}
  \country{}
}
\email{yuanxie@ust.hk}

\renewcommand{\shortauthors}{Zhaohui Yang et al.}

\begin{abstract}
  Despite remarkable milestones in quantum computing, the performance of current quantum hardware remains limited. One critical path to higher performance is to expand the quantum ISA with basis gates that have higher fidelity and greater synthesis capabilities than the standard $\mathrm{CNOT}$. However, this substantially increases gate calibration overhead and introduces challenges in compiler optimization. Consequently, although more expressive ISAs (even complex, continuous gate sets) have been proposed, they still remain primarily proofs-of-concept and have not been widely adopted.

  To move beyond these hurdles and unlock the performance gains offered by expressive continuous ISAs, we introduce the concept of ``reconfigurable quantum instruction set computers'' (ReQISC). It incorporates (1) a unified microarchitecture capable of directly implementing arbitrary 2Q gates equivalently, i.e., $\mathrm{SU}(4)$ modulo 1Q rotations, with theoretically optimal gate durations given any 2Q coupling Hamiltonian and (2) a compilation framework tailored to ReQISC primitives for end-to-end synthesis and optimization, comprising a program-aware pass that refines high-level representations, a program-agnostic pass for aggressive circuit-level optimization, and an $\mathrm{SU}(4)$-aware routing pass that minimizes hardware mapping overhead. 

  We detail the hardware implementation to demonstrate the feasibility of this superior gate scheme in terms of both pulse control and calibration.
  By leveraging the expressivity of $\mathrm{SU}(4)$ and the time minimality realized by the underlying microarchitecture, the $\mathrm{SU}(4)$-based ISA achieves remarkable performance, with a $4.97$-fold reduction in average pulse duration to implement arbitrary 2Q gates, compared to the usual $\mathrm{CNOT}/\mathrm{CZ}$ scheme on mainstream flux-tunable transmons. Supported by the end-to-end compiler, ReQISC outperforms the conventional $\mathrm{CNOT}$-based ISA, state-of-the-art compiler, and pulse implementation counterparts by significantly reducing 2Q gate count, circuit depth, pulse duration, qubit mapping overhead, and program fidelity losses. For the first time, ReQISC makes the theoretical benefits of continuous ISAs practically feasible.

\end{abstract}

\maketitle

\section{Introduction}\label{sec:Introduction}



\begin{figure}[tbp]
    \centering    
    \includegraphics[width=\columnwidth]{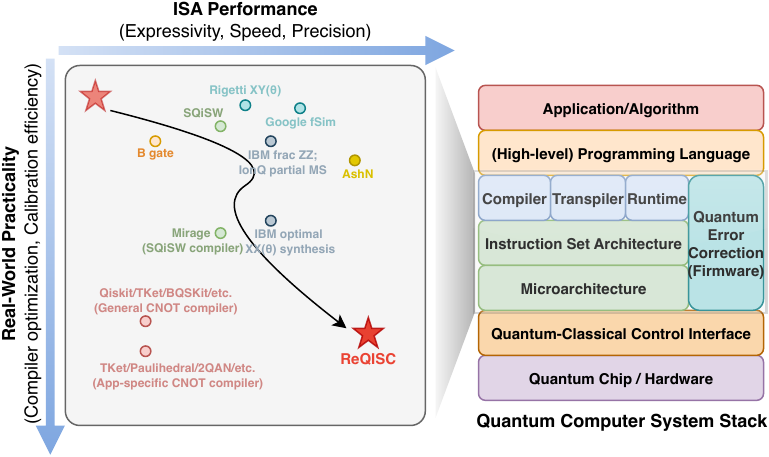}
    \caption{Advancing architectural support for practical quantum advantage.}

    \label{fig:motivation}
\end{figure}

Quantum hardware has advanced rapidly over the past two decades, achieving remarkable milestones including quantum supremacy~\cite{arute2019quantum}, quantum utility~\cite{kim2023evidence}, and achieving error rates below the fault-tolerant threshold~\cite{acharya2024quantum}. Today's quantum processors come with detailed device specifications such as native gate sets and pulse-level controls, accessible through a cloud interface or a complete system delivery. However, hardware noise presents critical challenges for both noisy intermediate-scale quantum (NISQ)~\cite{preskill2018quantum} and fault-tolerant application-scale quantum (FASQ)~\cite{preskill2025beyond} devices, significantly limiting executable program scales and impacting computational reliability~\cite{arute2019quantum,kim2023evidence,acharya2024quantum}.



To pursue a practical quantum advantage, many sophisticated architectural supports have been proposed by exploring advanced instruction set architectures (ISA) beyond the conventional $\CNOT$-based paradigm (\Cref{fig:motivation}).
For specific applications, novel 2Q quantum gates like Google's Sycamore from the $\fSim$ family offer higher fidelity and are harder to simulate classically~\cite{arute2019quantum}, while gates like $\SQiSW$ aim to improve fidelity and synthesis efficiency~\cite{huang2023quantum}. Physical platforms developed by Google, Rigetti, IonQ, Quantinuum, and IBM have also demonstrated continuous gate sets~\cite{foxen2020demonstrating,abrams2020implementation,ionqPartialGates,ibmFractionalGates,quantinuumArbitraryAngleGates},
which appear to yield improvements in fidelity~\cite{wright2019benchmarking,nam2020ground}. At the frontier, recent works have proposed using the entire $\mathbf{SU}(4)$ group as the most expressive ISA, achieving optimal-duration gates on flux-tunable transmons with $99.37\%$ average fidelity in experiments~\cite{chen2024one,chen2025efficient}.




Adopting more expressive quantum ISAs offers significant theoretical benefits, most notably a reduction in the 2Q gate count required for circuit synthesis~\cite{abrams2020implementation,foxen2020demonstrating,chen2024one,ionqPartialGates,ibmFractionalGates}. This is demonstrated by comparing the Haar-random synthesis cost of representative 2Q basis gates against the standard $ \CNOT $-equivalent gates. For instance, the cost is $ 3 $ for $ \CNOT/\CZ/\iSWAP $~\cite{shende2004minimal}, around $ 2.21 $ for $ \SQiSW $~\cite{huang2023quantum}, $ 2 $ for $ \B $ gate~\cite{zhang2004minimum}, and even lower for combinatorial or continuous gate sets~\cite{peterson2020fixed}. However, despite the variety of quantum ISAs available, none have successfully challenged the dominance of the $\CNOT$-based ISA due to several inherent challenges:

\textbf{Hardware implementation and calibration.} Implementing diverse 2Q gates with high precision and simple control via unified gate schemes across diverse quantum platforms remains a significant challenge, especially for complex, continuous quantum ISAs. For example, although $ \B $ gate was proposed decades ago~\cite{zhang2004minimum}, it has only recently been implemented natively~\cite{chen2025efficient,nguyen2024programmable}; IBM's fractional $\ZZ$ gates are deployed on their Heron QPUs~\cite{ibmFractionalGates}, but they have the same gate time as $\CZ$ and on average $1.5$ times the error rate of $\CZ$, according to public data from IBM Quantum Platform~\cite{ibmFractionalGates}. Furthermore, each 2Q gate must be carefully and periodically calibrated to ensure high precision. While proposals exist to mitigate calibration overhead~\cite{kelly2014optimal}, their effectiveness is not guaranteed. Thus, a program with numerous distinct 2Q gates presents significant execution challenges in practice.

\textbf{Compilation strategy.} Although the capability to implement more complex ISAs has been demonstrated, effectively utilizing them to compile programs into fewer gates with reduced circuit depth remains largely unaddressed. Neither gate set transpilation~\cite{peterson2022optimal,mckinney2024mirage,yale2024noise} nor numerically optimal synthesis~\cite{lao2021designing,kalloor2024quantum} effectively exploits their synthesis potential for real-world program compilation. We observe that ISA-customized strategies are crucial at every compilation stage, from high-level synthesis to hardware mapping. Otherwise, a theoretically superior ISA may be inferior in practice~\cite{kalloor2024quantum}.

To overcome these challenges and establish the practical superiority of expressive continuous ISAs, we introduce ReQISC (\underline{Re}configurable \underline{Q}uantum \underline{I}nstruction \underline{S}et \underline{C}omputers). ReQISC is a full-stack framework that integrates compilation, ISA design, and pulse control to boost performance across various quantum hardware platforms. At its core, ReQISC utilizes the most expressive $\mathbf{SU}(4)$ as its ISA (i.e., all 2Q gates) and makes it the highest-performing with minimal gate time through our proposed microarchitecture (gate scheme). The ReQISC microarchitecture can \emph{straightforwardly implement arbitrary 2Q gates up to local equivalence in optimal time, via simple pulse controls, given any 2Q coupling Hamiltonian}. Recent experimental progress has affirmed the viability of specific mechanisms in this direction~\cite{chen2025efficient}; our scheme, however, represents a more generalized approach with extensive architectural exploration, extending the advantages to a broader hardware spectrum. Our key contributions are:

\ding{182} We introduce a unified microarchitecture for the native, time-optimal realization of the full $\SU(4)$ group under arbitrary coupling Hamiltonians. We validate the pulse-control simplicity and calibration feasibility on representative hardware platforms. \note{To ensure experimental feasibility across all 2Q gates, we incorporate a gate mirroring mechanism that resolves the control singularity of near-identity gates at compile time without incurring 2Q gate count overhead.}

\ding{183} We present the first end-to-end compiler tailored to $ \SU(4) $ characteristics. Based on the lower-bound analysis of required $ \SU(4) $ resources for synthesizing multi-qubit circuits, our compiler employs both program-aware and program-agnostic passes to optimize the overall $ \SU(4) $ gate count. We address the program-specific calibration overhead and also effectively manage the trade-off between calibration overhead and aggressive 2Q gate count reduction.

\ding{184} We propose the program-aware template-based synthesis method that first pre-synthesizes optimal $ \SU(4) $-based circuit templates for refined intermediate representations (IR) and then selectively assembles the whole circuit.

\ding{185} We also propose a hierarchical, program-agnostic synthesis pass to aggressively lower the 2Q gate count, building on circuit partitioning and approximate synthesis. We introduce the partitioning compactness metric to guide co-optimization and develop a directed acyclic graph (DAG) compacting pass based on approximate commutation rules. 

\ding{186} Furthermore, we introduce a novel $\SU(4)$-aware routing pass, dubbed mirroring-SABRE, which significantly reduces the routing overhead induced by topology constraints.

\section{Background}\label{sec:Background}

\subsection{Quantum instruction set}


Similar to its classical counterpart, a quantum ISA serves as an interface between software and hardware, by mapping the high-level semantics of quantum programs to low-level native quantum operations or pulse sequences on hardware. Quantum instructions involve physical operations on qubits that alter their states. Typically, these instructions include qubit initialization, a universal gate set, and measurement. The universal gate set is the key component of a quantum ISA that dominates its hardware-implementation accuracy and cost, as well as software expressivity. It is well known that a gate set consisting of any 2Q entangling gate (e.g., $\CNOT$, $\iSWAP$) together with all possible 1Q gates can achieve universal computation~\cite{bremner2002practical}. Instructions from a universal gate set correspond to dynamics governed by the system Hamiltonian $H$, potentially influenced by external drives. The evolution of a quantum state over time $t$ is described by the unitary operator $U(t) = e^{-iHt}$ which realizes the desired quantum gate up to a global phase.



\subsection{Canonical decomposition and the Weyl chamber}

The group $\mathbf{SU}(2^n)$ is a real manifold with dimension $4^n - 1$; each point in this manifold corresponds to an $ n $-qubit unitary.
A generic 2Q gate within $ \mathbf{SU}(4) $, despite having $ 15 $ real parameters, can have its nonlocal behavior fully characterized by only $3$ real parameters. This method, known as Canonical decomposition or KAK decomposition from Lie algebra theory, is widely adopted in quantum computing~\cite{zhang2003geometric,tucci2005introduction,bullock2003arbitrary,zulehner2019compiling}.


Any $U \in \mathbf{SU}(4)$ can be defined by a unique vector $\Vec{\eta} = (x, y, z) \in W \subseteq \mathbb{R}^3$, along with $V_1, V_2, V_3, V_4 \in \mathbf{SU}(2)$, s.t.,
\begin{align}
U = g \cdot (V_1 \otimes V_2) e^{-i\Vec{\eta} \cdot \Vec{\Sigma}} (V_3 \otimes V_4),\, g \in \{1, i\}\label{eq:kak_decomposition}
\end{align}
where $\Vec{\Sigma} \equiv (XX, YY, ZZ)$~\cite{tucci2005introduction}. The set
\begin{align*}
W \coloneqq \{(x, y, z) \in \mathbb{R}^3 \,|\, \pi/4 \geq x \geq y \geq |z|,\, z \geq 0 \text{ if } x = \pi/4\}
\end{align*}
is known as the \emph{Weyl chamber}~\cite{zhang2003geometric}, and  $\Vec{\eta} \in W$ is known as the \emph{Weyl coordinate}
of $U$. We also refer to a gate of the form $\CanGate{x}{y}{z} \coloneqq e^{-i\Vec{\eta}\cdot\Vec{\Sigma}}$ as a \emph{canonical} gate. Two 2Q gates $ U $ and $ V $ are considered \emph{locally equivalent} if they differ only by 1Q gates, \note{and we write $ U \sim V $ to denote the equivalence relation.} For example, $ \CNOT\sim\CZ $, as $ \CZ=(I\otimes H)\,\CNOT\,(I\otimes H)$ and their canonical coordinates are both $ (\pi/4,0,0) $.

\begin{figure*}[tbp]
    \centering    
    \includegraphics[width=\linewidth]{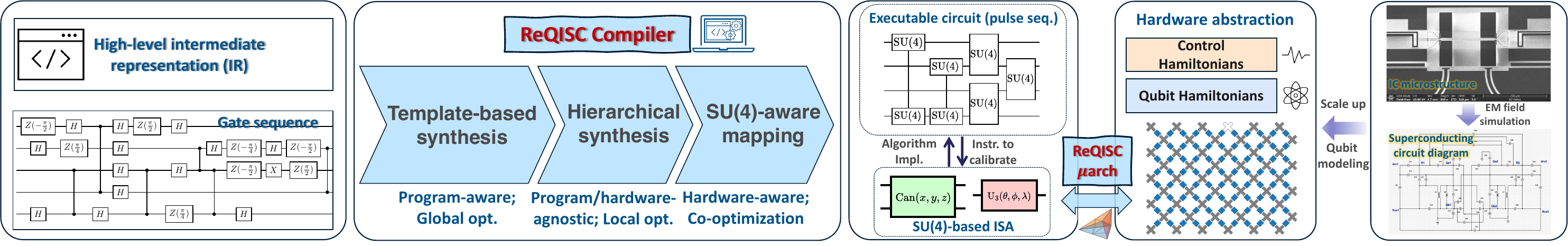}
    \caption{ReQISC workflow. It incorporates (1) the most performant (time-optimal realization) microarchitecture to engineer any arbitrary $\SU(4)$ given the hardware coupling Hamiltonian and (2) an end-to-end compilation framework with three-stage optimization passes to generate executable $\SU(4)$-based circuits.  The $ \SU(4) $ ISA is expressed as the $ \{\Can,\, \mathrm{U}3 \} $ gate set.}
    \label{fig:workflow}
\end{figure*}

\section{Our Proposal: ReQISC}\label{sec:ReQISC}

\Cref{fig:workflow} illustrates the overall workflow of ReQISC. The end-to-end ReQISC compiler first converts quantum programs in high-level IR or gate sequences in conventional (e.g., $\CNOT$-based) ISAs into $\SU(4)$ operations, as depicted by the circuit composed of $\CanGate{x}{y}{z} $ and $\UGate{\theta}{\phi}{\lambda}$ gates.
Subsequently, a set of simple pulse control parameters for desired $\SU(4)$ instructions is computed via the ReQISC microarchitecture and calibrated before executing on hardware. Note that we do not alter the hardware at any step. The electromagnetic simulation and qubit modeling follow standard procedures to specify a 2Q coupling Hamiltonian and tunable 1Q drive Hamiltonians by adjusting control parameters.

\begin{figure}[tbp]
    \centering
    \begin{subfigure}[t]{\columnwidth}
        \centering
        \includegraphics[width=\linewidth]{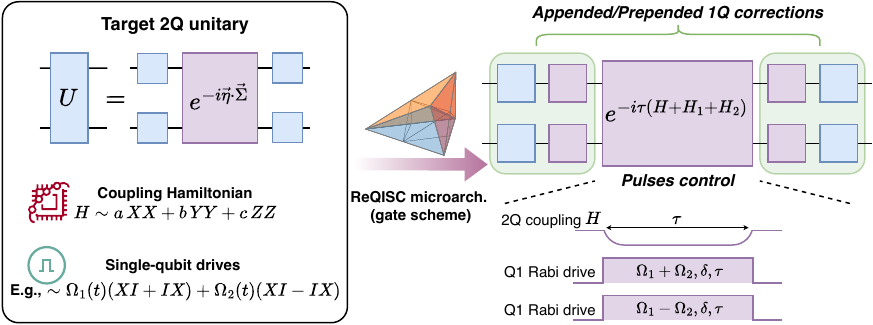}
        \caption{\note{Landscape of ReQISC microarchitecture (gate scheme)\\\quad}}
        \label{fig:gate_scheme:a}
    \end{subfigure}
    \begin{subfigure}[t]{\columnwidth}
        \centering
        \includegraphics[width=\linewidth]{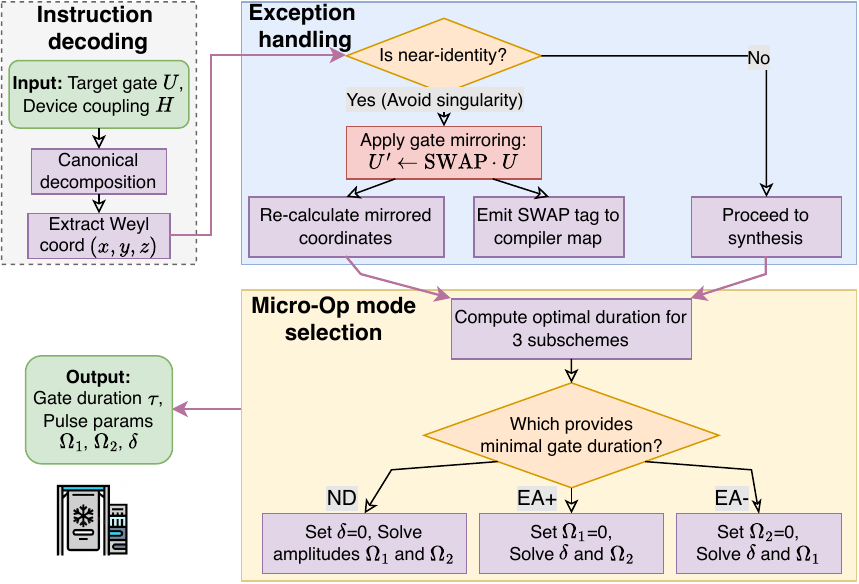}
        \caption{\note{ReQISC's micro-control logic from the system perspective}}
        \label{fig:gate_scheme:b}
    \end{subfigure}
    \caption{\note{ReQISC microarchitecture (gate scheme) directly implements arbitrary 2Q gates with a unified control. Common 2Q coupling Hamiltonians and standard 1Q Rabi drives with control parameters $\Omega_1$, $\Omega_2$, $\delta$, $\tau$, are sufficient to generate an evolution locally equivalent to the target unitary $U$.}} 
    \label{fig:gate_scheme}
\end{figure}

\section{The ReQISC microarchitecture}\label{sec:microarch}

While any 2Q gate can theoretically be realized through Hamiltonian steering~\cite{zhang2005generation,lin2022let,mckinney2023parallel}, pulse concatenation or optimization~\cite{song2019generation,omran2019generation,gokhale2020optimized}, they often require complex control sequences or fail to achieve time-optimality. In contrast, we propose a unified and simple control scheme that directly implements (up to local equivalence) arbitrary 2Q gates in optimal time, as depicted by \Cref{fig:gate_scheme} and \Cref{algo:isa}. Inspired by the AshN scheme~\cite{chen2024one}, which was limited to $\mathrm{XY}$ coupling, our scheme accommodates arbitrary coupling Hamiltonians for any quantum processor, whether it be superconducting, trapped-ion, or other platforms, significantly broadening applicability. 
Time optimality~\cite{hammerer2002characterization} refers to the theoretical lower bound of evolution time required to implement a locally equivalent 2Q unitary under a specific coupling Hamiltonian and tunable local drives. This is essential for mitigating decoherence and reducing overall runtime.

\note{


From a systems perspective, the hardware provides a fixed \dquote{datapath} determined by its physical coupling Hamiltonian (e.g., constant $\mathrm{XX}+\mathrm{YY}$ or $\mathrm{XX}$ interaction). ReQISC acts as the control logic that configures the local drives to steer this datapath. The generation of these control signals follows a rigorous logic flow, as formalized in \Cref{fig:gate_scheme}(b). It consists of three stages: \ding{172} \textbf{Instruction decoding:} The target gate $U$ is decomposed into its canonical coordinates $(x,y,z)$, mapping it to a point in the geometric instruction space (Weyl chamber). \ding{173} \textbf{Exception handling:} We identify \dquote{near-identity} $\SU(4)$ gates that physically require arbitrarily large energies to execute optimally due to control singularities. Instead of executing them directly, the compiler transforms these gates into their \dquote{mirrored} version by appending a logical SWAP, thereby avoiding the singularity. The effects of extra $\SWAP$ gates can be resolved by just tracking qubit mapping changes. \ding{174} \textbf{Micro-op synthesis:} Based on the instruction's coordinates, the logic selects one of three optimal execution modes---no-detuning (ND), equal-amplitude$_+$ (\EAplus), and equal-amplitude$_-$ (\EAminus). This selection is analogous to selecting the optimal datapath configuration in a CPU to minimize latency (gate duration $\tau$).


}




   \begin{algorithm}
              \SetAlgoLined
              \caption{ReQISC Microarchitecture Design}
              \label{algo:isa}
              \SetKwInOut{Input}{Input}
              \SetKwInOut{Output}{Output}
              \SetKwComment{Comment}{/*\scriptsize\,}{\,\scriptsize*/}
              \Input{Coupling Hamiltonian $H\in\mathcal{H}(\mathbb{C}^{4\times 4})$; Target $U\in \mathbf{SU}(4)$} 
              \Output{2Q gate time $\tau$; Local drive Hamiltonians $H_1, H_2\in\mathcal{H}({C^{2\times 2}})$; 1Q gate corrections $A_1,A_2,B_1,B_2 \in \mathbf{SU}(2)$}
      
              \BlankLine
              \SetKwBlock{Assumption}{Assumption}{}
          $(x,y,z,\, A'_1,A'_2,\, B'_1,B'_2)\gets \textsc{CanonicalDecompose}(U)$\;\label{line:cannonical-decompose}
          $(a,b,c,\, H'_1,H'_2,\, U_1,U_2)\gets \textsc{NormalForm}(H)$\;\label{line:normal-hamiltonian}


          $\tau_0,\, \tau_+,\, \tau_- \gets \frac{x}{a},\, \frac{x+y-z}{a+b-c},\, \frac{x + y + z}{a + b + c}$\;
          $\tau_1 \gets \max \{ \tau_0, \tau_+, \tau_-\}$\;
       
          $\tau_0',\,\tau_+',\,\tau_-'\gets\frac{\frac{\pi}{2} - x}{a},\,\frac{\frac{\pi}{2}- x + y + z}{a + b - c},\,\frac{\frac{\pi}{2}-x+y - z}{a+b+c}$\;
          $\tau_2 \gets \max\{ \tau_0', \tau_+', \tau_-'\}$\;
          $\tau \gets \min \{\tau_1, \tau_2\}$\; 
   
           \If{$\tau_2 < \tau_1$}{
               $x\gets \frac{\pi}{2}-x$;\quad $z\gets -z$\;
               $\tau_0,\, \tau_+,\, \tau_- \gets \tau_0',\, \tau_+',\, \tau_-'$\;
           }
           
          \eIf{$\tau = \tau_0$}{
              $S_1 \gets \mathrm{sinc}^{-1}\bigl(\frac{\sin(y - z)}{(b-c)\tau}\bigr) / \tau$ where $S_1\geq b-c$\;\label{line:sinc-S1}
              $S_2 \gets \mathrm{sinc}^{-1}\bigl(\frac{\sin(y + z)}{(b+c)\tau}\bigr) / \tau$ where $S_2\geq b+c$\;\label{line:sinc-S2}
              $\Omega_1\gets \frac{1}{2}\sqrt{S_1^2 - (b-c)^2};\, \Omega_2\gets \frac{1}{2}\sqrt{S_2^2 - (b+c)^2}\;
              \delta \gets 0$\;
          }{
               $\mathrm{lhs}(\alpha,\beta,\eta, t) \coloneqq \frac{(1-\alpha)(1+\alpha+\beta)-(1-\alpha+\beta)\eta}{(1-\alpha+\beta)(1-\eta+\alpha+2\beta)}e^{-i(2+2\beta-\eta) t}+\frac{\beta(1+\alpha+\beta)-(1-\alpha+\beta)\eta}{(1-\alpha+\beta)(2\alpha+\beta-\eta)}e^{i(\eta-2\alpha) t} + \frac{(1+\alpha-\eta)\eta - \beta(1-\alpha-\eta)}{(2\alpha+\beta-\eta)(1+\alpha+2\beta-\eta)}e^{i(2\alpha+2\beta-\eta) t}$\;\label{line:lhs}
               $\mathrm{rhs}(x,y,z) \coloneqq e^{i(x-y-z)}-e^{i(y-x-z)}+e^{i(z-x-y)}$;\label{line:rhs}

          \eIf{$\tau=\tau_+$}{
           $x',\,y',\,z' \gets x + c\tau,\, y + c\tau,\, c\tau - z$\;
           $\eta \gets (a-b) / (a + c);\quad t \gets (a+c)\tau$\;
            
            $(\alpha,\beta)\gets \textsc{NSolve}( \mathrm{lhs}(\alpha, \beta, \eta, t) = \mathrm{rhs}(x',y',z'))$ where $(\alpha, \beta)\in[0,1]\times [0,\infty)$ and $\alpha + \beta \geq \eta$\;\label{line:nsolve1}
            $\Omega_1\gets0 $; $\Omega_2\gets (a+c)\sqrt{(1-\alpha)\beta(1-\eta + \alpha +\beta)}$\;
            $\delta\gets -(a+c)\sqrt{\alpha(1+\beta)(\alpha + \beta - \eta)}$\;
          }{
            $x',\,y',\,z' \gets x - c\tau,\, y - c\tau,\, z - c\tau$\;
           $\eta \gets (a-b) / (a - c);\quad t \gets (a-c)\tau$\;
            $(\alpha,\beta)\gets \textsc{NSolve}( \mathrm{lhs}(\alpha, \beta, \eta, t) = \mathrm{rhs}(x',y',z'))$ where $(\alpha, \beta)\in[0,1]\times [0,\infty)$ and $\alpha + \beta \geq \eta$\;\label{line:nsolve2}
            $\Omega_1 \gets (a-c) \sqrt{(1-\alpha)\beta(1-\eta + \alpha +\beta)}$; $\Omega_2 \gets 0$\;
            $\delta = (a-c) \sqrt{\alpha(1+\beta)(\alpha + \beta - \eta)}$\;
            
            
      }
          }
          $H''_1,\, H''_2\gets (\Omega_1+\Omega_2)X + \delta Z,\, (\Omega_1-\Omega_2)X + \delta Z$\;
          $(x',y',z',\,A_1'',A_2'',\,B_1'',B_2'')\gets \textsc{CanonicalDecompose}(e^{-i(aXX+bYY+cZZ + H''_1 + H''_2)\tau})$\Comment*[r]{$(x',y',z')=(x,y,z)$}
          $H_1,\,H_2\gets U_1 H''_1 U_1^\dagger-H'_1,\, U_2 H''_2 U_2^\dagger-H'_2$\;
          $A_1,\,A_2\gets A'_1 (A''_1)^\dagger U_1^\dagger, A'_2 (A''_2)^\dagger U_2^\dagger$\;\label{line:local_corrections_after}
          $B_1,\, B_2\gets U_1 (B''_1)^\dagger B'_1,\, U_2 (B''_2)^\dagger B'_2$\Comment*[r]{$(A_1\otimes A_2)e^{-i\tau(H+H_1\otimes I + I\otimes H_2)}(B_1\otimes B_2) = U$}\label{line:local_corrections_before}
   
   \end{algorithm}


\subsection{Protocol design (Hamiltonian description)}\label{sec:protocol-design}
{
To illustrate the general applicability of our scheme, we consider the most general 2Q coupling Hamiltonian $H$ and local Rabi drives $H_{d_i}$~\cite{zhang2005generation} for a pair of resonant qubits ($\omega_{q_1}=\omega_{q_2}=\omega$), with experimentally tunable pulse parameters $A_i(t), \delta$:
\begin{align}
    H = a\,XX + b\,YY + c\,ZZ ;\, H_{d_i}=-\frac{1}{2}A_i(t)X_i + \delta\, Z \label{eq:normalized_hamiltonian}
\end{align}
where $(a, b, c)$ are the canonical coupling coefficients~\cite{bennett2002optimal} with $ a \geq b \geq \lvert c \rvert $. 
We define the coupling strength \note{
\begin{align}
    g :=a+b+\lvert c\rvert
\end{align}
to facilitate the comparison of different physical platforms.}
For flux-tunable superconducting qubits, in the drive frame after the rotating wave approximation, $a=b=g/2$, $c=0$, $A_i(t)$ is the amplitude of a square wave envelope for a sinusoidal microwave drive, which we will assume is constant, and $\delta_i:=\frac{1}{2}(\omega_{d_i} - \omega)$ is the drive detuning (see~\cite{krantz2019quantum,chen2024one} for detailed derivations). We show that it is sufficient to set the drive frequencies to be identical, so $\delta_1 = \delta_2 = \delta$. For other physical platforms, $H_{d_i}$ may be realized via other means (e.g., lasers in Rydberg atoms or trapped ions). In this case, $A_i$ and $\delta_i$ may correspond to other physical parameters. Our framework can also be easily adapted to other possible $H_{d_i}$ (for example a $Y$ term).
For mathematical convenience, we define $ \Omega_{1,\,2} \coloneqq -\frac{1}{4}{(A_1 \pm A_2)} $, and we can rewrite the local Hamiltonians, which now act on the 2Q Hilbert space, as
\begin{align}
    H_1 := (\Omega_1+\Omega_2)\, XI+ \delta\, ZI,\, H_2 := (\Omega_1-\Omega_2)\, IX +\delta \, IZ.
\end{align}
The local drive amplitudes $\Omega_1$ and $\Omega_2$, detuning $\delta$, and interaction duration $\tau$ constitute the set of simple control parameters to be determined via our microarchitecture such that simultaneously applying $H_1,H_2$ with the device 2Q coupling $H$ for duration $\tau$ implements a gate locally equivalent to a target gate $U$, as illustrated in \Cref{fig:gate_scheme}.

}


\Cref{algo:isa} describes the details of our gate scheme. A target gate $U$ and the device 2Q Hamiltonian $H$ are provided as inputs. The goal is to determine the local Hamiltonians $H_1$ and $H_2$, with the interaction duration $\tau$, such that $e^{-i\tau (H + H_1 + H_2)}\sim U$. Specifically, the Weyl coordinates $(x,y,z)$ of $U$ are obtained via the KAK decomposition (line \ref{line:cannonical-decompose}) according to \Cref{eq:kak_decomposition}. 
\note{The interaction Hamiltonian $ H $, if not already in normal form, is transformed into its normal form~\cite{dur2001entanglement} to obtain the coefficients $ (a,b,c) $ (line \ref{line:normal-hamiltonian}):
\begin{align*}
    H = (U_1\otimes U_2) (a\,XX+b\,YY+c\,ZZ)(U_1^\dagger\otimes U_2^\dagger) + H'_1  + H'_2,
\end{align*}
where $U_1, U_2$ are 1Q unitaries, $H_1', H_2'$ are 1Q Hermitian operators, and $a\geq b\geq |c|$.}
The parameters $(x,y,z)$ and $(a,b,c)$ are then fed to the subsequent solving procedure that computes $H_1, H_2, \tau$. With additional 1Q gate corrections $(A_1, A_2)$ and $(B_1, B_2)$ (line \ref{line:local_corrections_after} and \ref{line:local_corrections_before}), finally
\begin{align}
    (A_1\otimes A_2)\, e^{-i\tau(H+H_1 + H_2)}(B_1\otimes B_2)
\end{align}
exactly implements the desired $U$.


\note{Herein we give a high-level explanation of \Cref{algo:isa} while providing the detailed proof in Appendix \ref{appendix:proof}.}
First, the interaction time $\tau$ is determined as the maximum value among $\{ \tau_0, \tau_+, \tau_- \}$. Which of the three is largest dictates the subscheme used:
\note{\begin{enumerate}
    \item $\tau_0$ for no-detuning (ND, $\delta = 0$).
    \item $\tau_+$ for equal-amplitude with opposite signs (\EAplus, $\Omega_1 = 0$), i.e., the sinusoidal drive amplitudes applied on qubit 1 and qubit 2 have equal magnitude but opposite signs.
    \item $\tau_-$ for equal-amplitude with same sign (\EAminus, $\Omega_2 = 0$). 
\end{enumerate}}
\noindent The three subschemes correspond to three distinct regions that partition the Weyl chamber~\cite{chen2024one}. 
This partitioning of the Weyl chamber is convenient in terms of physical implementation but not unique. 


\note{
The proof of the correctness of~\Cref{algo:isa} starts with transforming into the Bell (or \dquote{magic}) basis, where the canonical coupling Hamiltonian $aXX+bYY+cZZ$ becomes diagonal. In this basis, applying local drives effectively mixes different eigenstates. As explained in Appendix \ref{appendix:proof}, using the Bell basis motivates analyzing the unitary 
\begin{align}
    V \coloneqq e^{-i\tau(H+H_1+H_2)} \cdot YY    
\end{align}
whose eigenvalues are related to the Weyl coordinates of the gate that our gate scheme realizes. We want these eigenvalues to match the desired Weyl coordinates $(x,y,z)$ of $U$. Specifically, the \dquote{ND} subscheme leads to a special class of Hamiltonian for which $V$ has a simple analytic form. In contrast, the \dquote{\EAplus/\EAminus} subschemes involve equations whose solutions do not have a simple analytic form, as written in \Cref{algo:isa}.  Furthermore, since the chosen duration $\tau$ matches the lower bound in~\cite{hammerer2002characterization}, our gate scheme is time-optimal.
}

\note{
\subsection{Numerical solving}\label{sec:numerical_solving}
}

\begin{figure}
    \centering
    \includegraphics[width=\columnwidth]{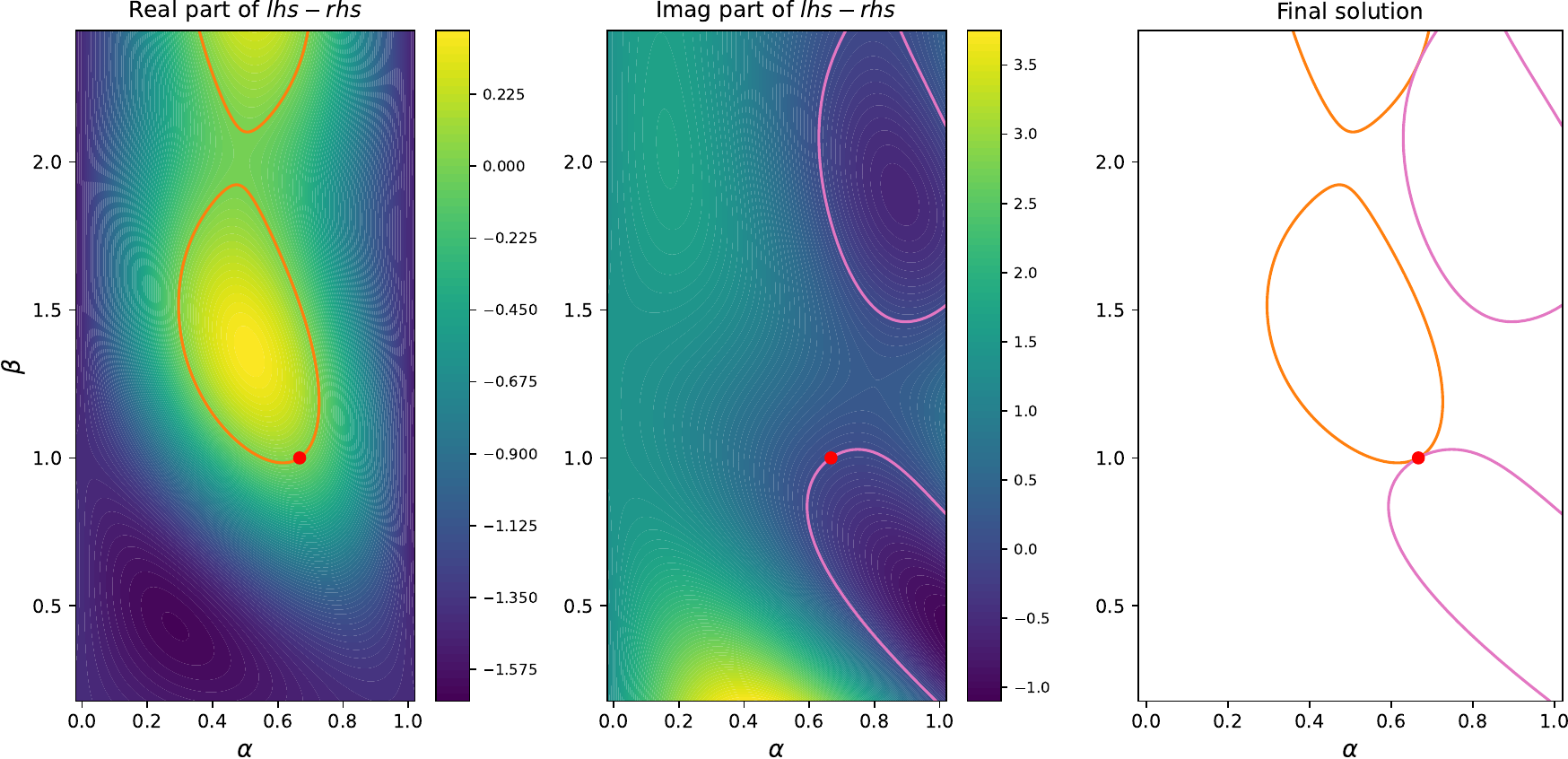}
    \caption{\note{$ (\alpha,\, \beta) $ solution profiling for $ \SWAP $ gate implementation under $ \mathrm{XX} $ coupling ($a\neq 0,\, b = c = 0 $). All intersection points of orange and purple lines are valid solutions to the transcendental equations (\dquote{\EAplus} subscheme).
    Our numerical solver finds the optimal solution for hardware implementation (red point in the rightmost subfigure).}}
    \label{fig:profiling_solving}
\end{figure}

\begin{figure}[tbp]
    \centering

    \includegraphics[width=\columnwidth]{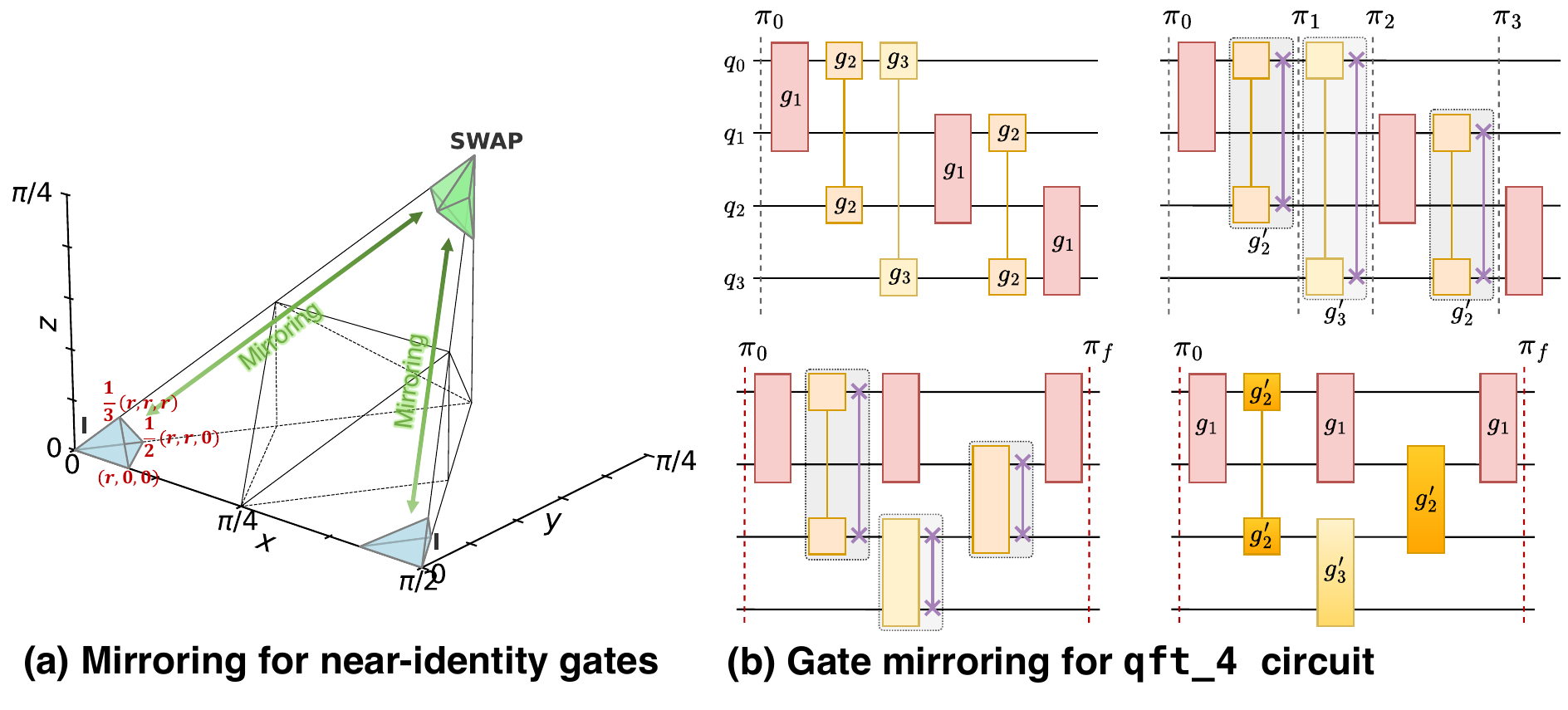}

    \caption{\note{Compile-time singularity resolution via gate mirroring. (a) Gates in the near-identity region ($ \mathcal{L}^1 \leq r$) are mirrored to the $\SWAP$ corner. (b) Application to \code{qft\_4}: In this example, $ g_i\sim \CanGate{\pi/{2^{2+i}}}{0}{0} $, such that $ g_2 $ and $ g_3 $ are assumed to be close enough to identity and are replaced by their mirrors while intermediate mappings are updated. Ultimately this resolution requires only one final mapping update $ \pi_f $ and introduces no extra 2Q gate.}}
    \label{fig:mirroring_weyl}

\end{figure}

Each subscheme involves transcendental equations that lack analytical solutions and must be solved numerically. 
Furthermore, since $\mathrm{sinc}(x)$ function (lines \ref{line:sinc-S1} and \ref{line:sinc-S2}) is not bijective, $ S_1 $ and $ S_2 $ have multiple solutions. To maximize the fidelity of a hardware implementation, we select by default the smallest solutions so that $\max\{\vert A_1\vert, \vert A_2\vert \}$ is minimized. 
\note{While determining the primary roots for $S_1$ and $S_2$ is straightforward, the numerical solver for the \EAplus/\EAminus subschemes requires a more sophisticated approach. We employ a hybrid numerical optimization strategy that integrates a coarse grid search with a two-stage local refinement to solve for $ (\alpha,\,\beta) $: \ding{172} \textbf{Grid search initialization}: We iterate over a predefined grid of initial guesses for $\alpha$ and $\beta$ both linearly spaced to adequately sample the parameter space. \ding{173} \textbf{Two-stage refinement}: For each $(\alpha,\,\beta)$ candidate in the grid search, we utilize the \code{SLSQP} method via \code{scipy.optimize.minimize} to minimize the residual between the expressions \dquote{$ \mathrm{lhs} $} (line \ref{line:lhs}) and \dquote{$ \mathrm{rhs} $} (line \ref{line:rhs}) but is of limited precision. The result is then refined to pinpoint the exact root using \code{scipy.optimize.fsolve}.
\ding{174} \textbf{Optimal selection}: Among the valid solutions, we select $ (\alpha,\,\beta) $ that minimizes a \dquote{physical implementation penalty}, defined as the sum of the absolute values of the resulting pulse amplitudes $ \Omega_1 $ or $ \Omega_2 $ and detuning $ \delta $.

This solving strategy achieves exceptional numerical precision
and robustness. The Weyl coordinates computed from the returned solution has a mean absolute error on the order of $10^{-16}$ for ND and $10^{-13}$ for \EAplus/\EAminus. The resulting unitary produced has an infidelity on the order of $10^{-15}$, effectively reaching the limit of machine precision.
In our field test with typical coupling Hamiltonians and common 2Q gates, as well as millions of random coupling Hamiltonians and target unitaries, the solving procedure consistently converges to high-precision solutions.
\Cref{fig:profiling_solving} illustrates the robustness of this procedure in identifying the primary $(\alpha, \beta)$ solution that corresponds to minimal pulse amplitudes and detuning for the SWAP gate under XX coupling.}

\subsection{Addressing the \dquote{near-identity} issue}\label{sec:gate_mirroring}


Following ~\Cref{algo:isa}, realizing gates with Weyl coordinates close to the origin (common in quantum Fourier transform (QFT) and Hamiltonian simulation kernels) in optimal time requires unbounded pulse amplitudes, which are experimentally infeasible. While trading gate time for bounded amplitudes is possible~\cite{chen2024one}, the longer gate time is still undesirable. To address this, we propose transforming the near-identity gate into its mirror gate~\cite{cross2019validating} which differs by a $\SWAP$ gate:
\begin{align*}
    \mathrm{SWAP}\cdot\mathrm{Can}(x,y,z)
    \sim \begin{cases}
    \mathrm{Can}\left(\frac{\pi}{4}-z, \frac{\pi}{4}-y, x - \frac{\pi}{4}\right), & \textrm{if } z\geq 0 \\
    \mathrm{Can}\left(\frac{\pi}{4} + z, \frac{\pi}{4}-y, \frac{\pi}{4} - x\right), & \textrm{if } z < 0
    \end{cases},
\end{align*}
and subsequently altering qubit mapping at compile time.
It is clear that the Weyl coordinates of the mirrors of near-identity gates are far from the origin, \note{as depicted in \Cref{fig:mirroring_weyl}(a)}.
During compilation, we append a $\SWAP$ to each near-identity 2Q gate, leaving the circuit rewiring effects of these $\SWAP$s to be resolved at compile time. Specifically, we define a coordinate norm threshold to determine if a 2Q gate is near-identity and requires mirroring. The threshold $r$ in general depends on hardware parameters~\cite{chen2024one}. The impacts of inserted $\SWAP$s are tracked by recording the altered qubit mappings. Thus we resolve the issue at compile time without necessarily introducing additional 2Q gate count overhead. For example, \note{\Cref{fig:mirroring_weyl}(c) illustrates an example for the four-qubit QFT circuit, which eventually leads to a circuit with no extra 2Q gates and only an updated final mapping to track.}


\subsection{Hardware implementation}\label{sec:hardware-implementation}

\begin{figure*}[tbp]
    \centering    
    \includegraphics[width=\textwidth]{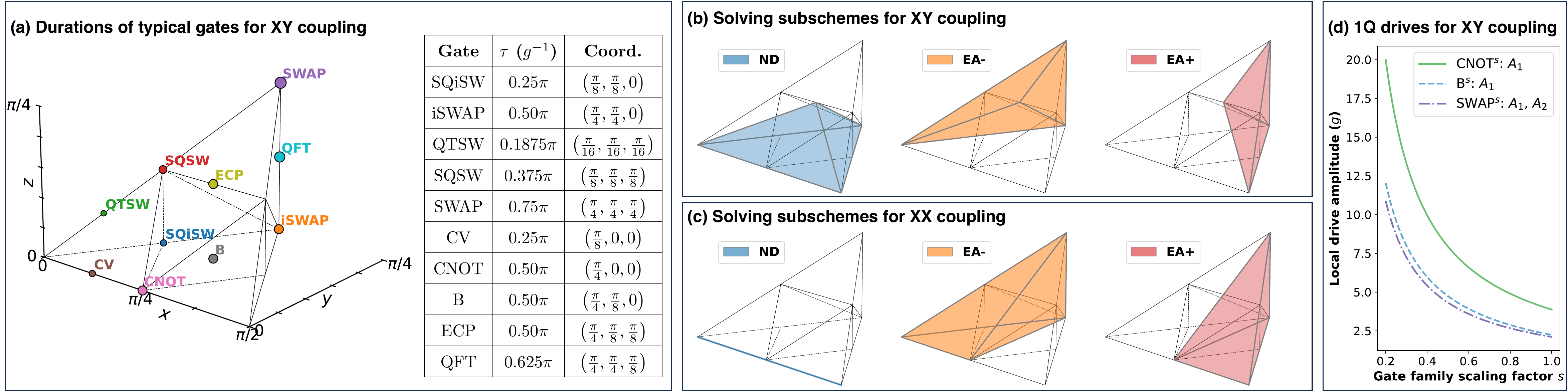}
    \caption{{Hardware implementation of the ReQISC microarchitecture. (a) Gate time landscape under $\mathrm{XY}$ coupling ($H=\frac{g}{2}(XX+YY)$); circle size indicates gate duration. (b-c): Pulse control subschemes, under $\mathrm{XY}$ and $ \mathrm{XX} $ couplings, respectively. (d) Local drive amplitudes required for representative gate families. The scaling factor $s$ is used to take fractions of gates. For example, $\mathrm{B}^s \sim \CanGate{\frac{\pi}{4}s}{\frac{\pi}{8}s}{0}$. The $\iSWAP$ family requires no local drives so it is not shown in the figure; the $\CNOT$ family and $\B$ family require only one-side drive ($A_1\neq 0$); $\SWAP$ family requires both-side drives ($A_1= A_2 \neq 0$). Frequency-related quantities ($A_1$, $A_2$, $\delta$) are normalized by the coupling strength $g$.}}
    \label{fig:solving-schemes}
\end{figure*}

To implement our presented microarchitecture on a hardware platform, the only requirement is that $H,H_1,H_2$ can be applied simultaneously. This is true for the case of flux-tunable superconducting qubits: standard microwave drives, arbitrary wave generators, and mixers are sufficient, as recently shown in an experiment~\cite{chen2025efficient}. Given these minimal requirements, many platforms, e.g., flux-tunable transmons~\cite{arute2019quantum,acharya2024quantum}, fluxonium qubits~\cite{nguyen2019high,bao2022fluxonium}, Floquet qubits~\cite{nguyen2024programmable}, semiconductor spin qubits~\cite{dijkema2025cavity}, trapped ions~\cite{yale2025realization}, and some neutral-atom devices~\cite{barredo2015coherent} could immediately benefit from this scheme. 

We outline how to realize common 2Q gates ($\CNOT$, $\iSWAP$, $\B$, etc.) with two typical coupling Hamiltonians: $\mathrm{XY}$ coupling ($H = \frac{g}{2}\,(XX + YY)$) and $ \mathrm{XX} $ coupling ($H = g\, XX$). For instance, the former is common in flux-tunable superconducting qubits~\cite{krantz2019quantum,arute2019quantum}; the latter is the dominant coupling type for trapped ions~\cite{yale2025realization}.
The latter type of coupling can also appear in capacitively coupled flux-tunable transmons by working in the lab frame instead of the drive frame~\cite{krantz2019quantum}:
\begin{align}
    H' = -\frac{\omega_1}{2} ZI - \frac{\omega_2}{2} IZ + gXX + V_1(t) XI + V_2(t)IX,    
\end{align}
where $\omega_i$ are the frequencies of the qubits, $g$ is the coupling strength, and $V_i(t)$ are the local drives. 
We can simply tune different qubit frequencies to control the $ZI, IZ$ terms while local drives can be square pulses with variable amplitudes instead of sinusoidal pulses to control the $XI, IX$ terms. This Hamiltonian is not supported by previous work~\cite{chen2024one}, yet can arise for certain hardware devices. Our gate scheme can handle this or any other coupling Hamiltonian.

Through \Cref{algo:isa}, we can derive the subschemes (ND, \EAplus, or \EAminus) and control parameters as shown in \Cref{fig:solving-schemes}(b-d). With respect to our gate scheme, there are some takeaways for ISA design. For example, gates requiring no or low amplitude drives can potentially achieve higher fidelities than those requiring stronger drives (e.g., $\CNOT$ under $\mathrm{XX}$ coupling vs. $\CNOT$ under $\mathrm{XY}$ coupling, $\iSWAP$ under $\mathrm{XY}$ coupling vs. $\iSWAP$ under $\mathrm{XX}$ coupling, respectively). This is because precise calibration of local drives with the coupling turned on is nontrivial in practice. Using our gate scheme, $\SWAP$ is far less expensive than the conventional three-$\CNOT$ decomposition, as shown for $\mathrm{XY}$ coupling in \Cref{fig:solving-schemes}(b). Even the $\CNOT$ implementation via our scheme is $1.41$x faster than the traditional control scheme on superconducting processors~\cite{krantz2019quantum}, that is, $\pi/2g$ vs. $\pi/\sqrt{2}g$. This speedup is critical for suppressing errors on platforms where decoherence is the dominant error source. Besides, our scheme does not introduce the additional $\lvert 11 \rangle \leftrightarrow \lvert 02\rangle$ transition, thereby minimizing leakage error and dynamic stray $\mathrm{ZZ}$ crosstalk. These advantages were demonstrated on transmons in~\cite{chen2025efficient}, where various 2Q gates are implemented in high fidelity (on average $ 99.37 $\%), with optimal gate duration. As our scheme represents a general approach to directly implement arbitrary $\SU(4)$ on physical platforms with arbitrary coupling Hamiltonians, it provides more opportunities for hardware-native ISA design and compiler optimization.



\subsection{Gate calibration}\label{sec:microarch-calibration}
Gate calibration can follow a similar procedure as outlined in~\cite{chen2025efficient}, where $\mathrm{XY}$ coupling is assumed. First, the $\iSWAP$-family component of the Hamiltonian (coupling term $H$) and the drive components ($H_1$ and $H_2$) are separately calibrated to determine how the model parameters $g$, $\Omega_{1,\,2}$, and $\delta$ depend on physical control parameters. This provides an initial estimate for physical control parameters. Then, the two parts of the Hamiltonian are simultaneously applied. Further optimization is achieved by first applying quantum process tomography~\cite{greenbaum2015introduction}. This is used to measure the Weyl coordinate of the physical gate realized. Control parameters are tuned to minimize the Euclidean distance from target coordinates. Further fine-tuning of the parameters can be achieved via cross-entropy benchmarking~\cite{arute2019quantum}.

Such a calibration method is readily generalizable. For example, instead of calibrating an $\mathrm{iSWAP}$-family gate, one would calibrate an $\mathrm{XX}(\theta)$ rotation for the case of $\mathrm{XX}$ coupling. To calibrate a continuous gate set, the idea in~\cite{chen2024one} could be utilized. In this case, we can calibrate the \emph{mapping} between the control parameters and model parameters. Such a mapping can be expressed as some parameterized mathematical function. We can then optimize the parameters of this function via fully randomized benchmarking (Haar-random 2Q gates)~\cite{kong2021framework} or other more general benchmarking schemes~\cite{chen2022randomized} with respect to different gate distributions. 








\section{ReQISC Compilation Framework}\label{sec:Compiler}

Using the ReQISC microarchitecture, we can natively realize the $\SU(4)$-based ISA with theoretically optimal performance. However, a dedicated compiler is still required to fully unlock its potential, as compiling generic programs into 2Q basis gates is non-trivial. In this section, we delve into the ReQISC compiler design, which draws upon novel insights regarding the analysis of hardware primitives, the characteristics of quantum programs, and hardware requirements.

\subsection{Local optimization: Hierarchical synthesis}\label{sec:compiler-hierarchical-synthesis}

We first attempt to exploit $\SU(4)$'s synthesis potential based on its resource-bound analysis and devise a \emph{hierarchical synthesis} structure to perform aggressive local optimization at the circuit level. 
We introduce a \dquote{compactness} metric across local subcircuits to guide co-optimization and devise a supporting \emph{DAG compacting} pass with the \emph{approximate commutation} rules in our findings.

\subsubsection{Resource analysis and approximate synthesis with $\mathbf{SU}(4)$}

The lower bound of 2Q basis gate count given a quantum ISA typically reflects its intrinsic synthesis power.  The bound for an $\SU(4)$-based ISA is $b_{\SU(4)}(n) = \lceil(4^n-3n-1)/9\rceil$, implying a theoretical 55.6\% reduction in \#2Q gates compared to a $\CNOT$-based ISA, whose bound is $b_{\CNOT}(n) =\lceil(4^n-3n-1)/4\rceil$. While no constructive algorithms can achieve this bound, the \emph{approximate synthesis} technique has demonstrated this capability~\cite{khatri2019quantum,madden2022best}. Through structural optimization of gate arrangements and numerical optimization of gate parameters, approximate synthesis can generate circuits with fewer gates that are equivalent to a target unitary within a permissible precision~\cite{davis2019heuristics}. This precision is typically set to be negligible ($10^{-15}$-$10^{-10}$, measured by infidelity $1 - \frac{1}{N} \lvert \mathrm{Tr}(U^\dagger V) \rvert$) compared to physical errors, making the synthesis effectively exact for practical purposes.

\begin{figure}[tbp]
    \centering
    \begin{subfigure}[t]{0.67\columnwidth}
        \centering
        \includegraphics[width=\linewidth]{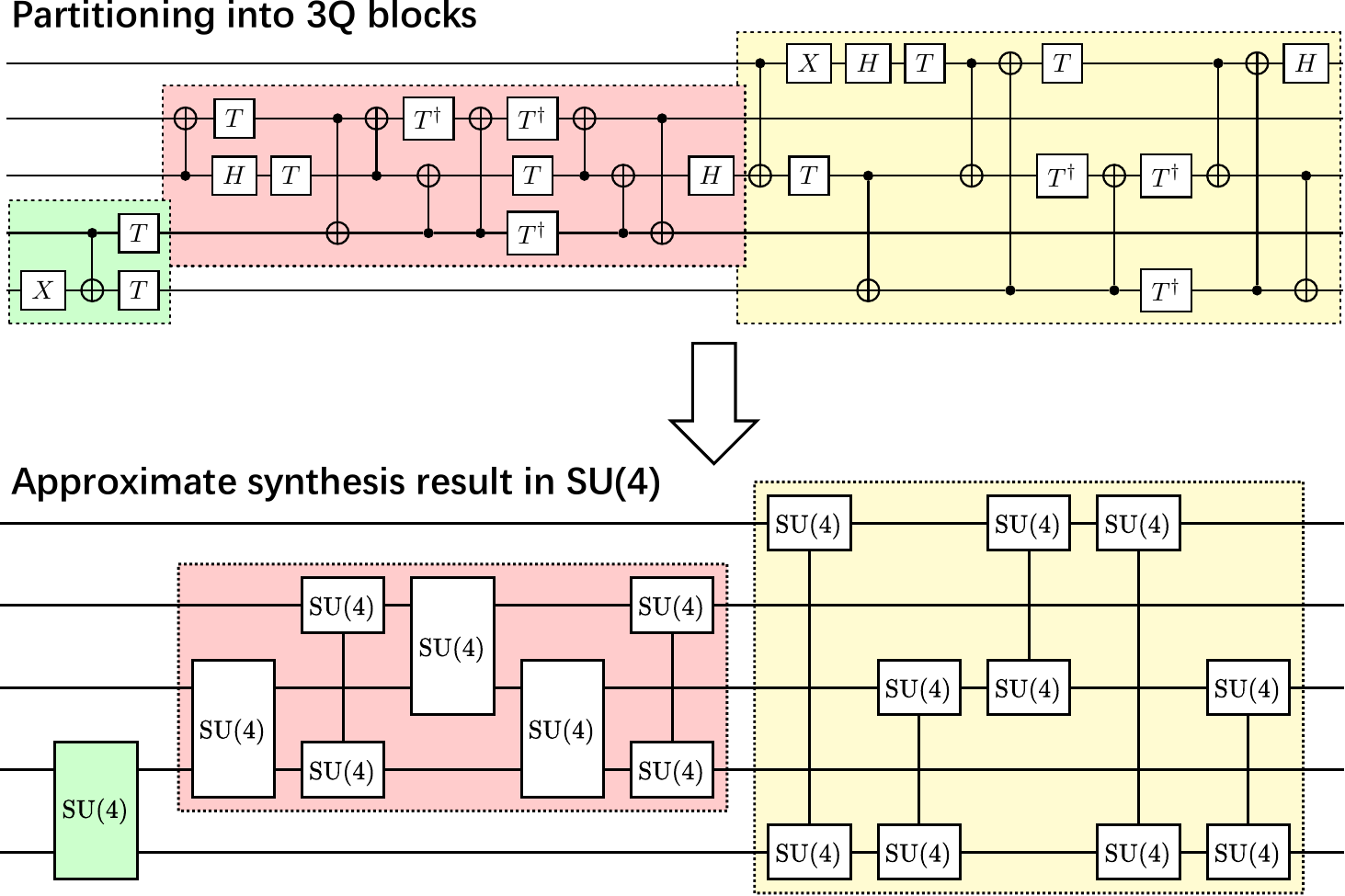}
        \caption{HS example for \code{alu-v2\_33} circuit}
        \label{fig:hierarchical_synthesis:a}
    \end{subfigure}
    \hfill
    \begin{subfigure}[t]{0.26\columnwidth}
        \centering
        \includegraphics[width=\linewidth]{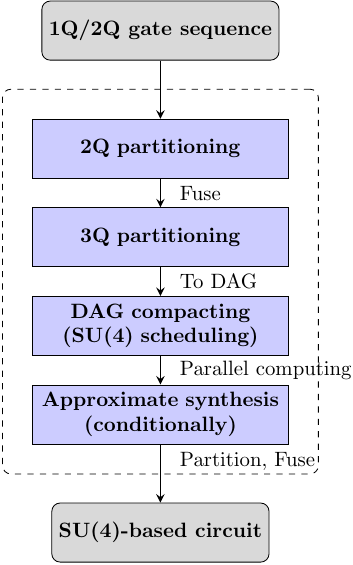}
        \caption{HS pipeline}
        \label{fig:hierarchical_synthesis:b}
    \end{subfigure}
    \caption{Example and pipeline of hierarchical synthesis (HS). (a) The original \code{alu-v2\_33} circuit is partitioned into 3 three-qubit blocks with \#2Q equal to 1, 8, and 8. After approximate synthesis on the latter two blocks, each has 5 $ \SU(4) $s. The overall \#2Q is reduced from 17 to 11. (b) HS involves two-tier partitioning, followed by DAG compacting and conditional approximate synthesis.}
    \label{fig:hierarchical_synthesis}
\end{figure}


\begin{figure}[tbp]
    \centering
    \includegraphics[width=\columnwidth]{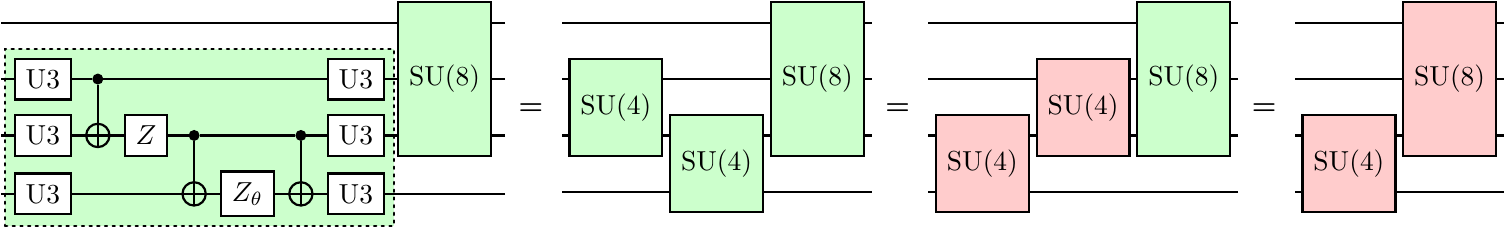}
    \caption{Exchanging \dquote{approximately commutative} $\SU(4)$s improves compactness. This pattern occurs in \code{grover\_5\_0}. Different colors indicate that parameters of $\SU(4)$ or $\SU(8)$ have changed.}
    \label{fig:approximate-commutation}
\end{figure}

\subsubsection{Hierarchical synthesis}\label{sec:compiler-hierarchical-synthesis-hyperparameter}

To leverage the theoretical synthesis potential of $\SU(4)$ and apply approximate synthesis to large-scale real-world programs, we employ a two-tier hierarchical synthesis approach. It first partitions the original circuit into 2Q blocks, each of which is fused into an $\SU(4)$ gate. Then, it partitions the circuit, now consisting solely of $\SU(4)$ gates, into $w$-qubit blocks to determine if each block can be synthesized with fewer $\SU(4)$ gates through approximate synthesis. Specifically, our compiler conditionally performs approximate synthesis on blocks whose $\SU(4)$ count is greater than a threshold $m_{th}$. Thus, two hyperparameters need to be carefully considered to strike a trade-off between computational overhead and 2Q gate count reduction effect via approximate synthesis: partitioning granularity $w$ and 2Q gate count threshold $m_{th}$. We determine optimal default values through empirical analysis: \ding{172} \textbf{Partition granularity ($w=3$):} Approximate synthesis is computationally expensive, scaling exponentially with circuit width. We found that 3Q partitioning offers an optimal balance, as 3Q subcircuits frequently contain more gates than the theoretical minimum ($b_{\SU(4)}(3) = 6$), providing ample optimization opportunities. In contrast, 4Q partitioning results in excessive computational overhead for the underlying approximate synthesis, and its potential for reduction is less frequently realized in practice, as $ b_{\SU(4)}(4) = 27 $ is too large. \ding{173} \textbf{Synthesis threshold ($m_{th}=4$):} In principle, $m_{th}$ can be set to the \#2Q lower bound, i.e., $ 6 $. However, empirical observations indicate that as few as $ 5 $ or even $ 4 $ gates are often sufficient to approximate 3Q blocks within the permissible precision, suggesting that a more aggressive $m_{th}$ can yield greater gate count reductions. As shown in \Cref{fig:hierarchical_synthesis}, this method reduces the \#2Q for the \code{alu-v2\_33} circuit from $ 17 $ to $11$.

\subsubsection{Compactness and DAG compacting}
 
To assess the circuit partitioning effects, we introduce the metric of \emph{compactness} to guide co-optimization of circuit partitioning and subsequent approximate synthesis: an ideal partition creates highly \dquote{unbalanced} blocks, concentrating many gates into a few blocks ripe for synthesis while leaving others sparse. Specifically, we desire that in the second-stage partitioned result, $w$-qubit blocks with $\mathrm{\#2Q} > m_{th}$ concentrate more 2Q gates, while those with $\mathrm{\#2Q} \leq m_{th}$ include fewer. 

Under the guidance of the compactness metric and the gate scheduling idea, we propose the \emph{DAG compacting} pass to conduct aggressive optimization for partitioned $\SU(4)$-based subcircuits that comprise a nested DAG. The DAG compacting pass exploits opportunities of exchanging \dquote{commutative} $\SU(4)$s to enhance compactness. For example, a gate sequence like $[\SU(4)_{1,2},\, \SU(4)_{2,3}]$ can often be exchanged into $[\SU(4)'_{2,3},\, \SU(4)'_{1,2}]$ with negligible error (\Cref{fig:approximate-commutation}). By performing this commutation, $\SU(4)'_{1,2}$ can merge with the subsequent 3Q block $\SU(8)_{0,1,2}$, thereby increasing compactness. The pass iteratively applies these feasible exchanges to maximize the compactness-increasing opportunities, enabling more aggressive optimization.



\subsection{Global optimization: Program-aware synthesis}

The hierarchical synthesis works at the circuit level and will induce diverse parameterized quantum gates, leading to uncontrollable calibration overhead. It is impractical for large-scale program compilation on its own. Thus, an approach to globally optimizing quantum programs and managing the calibration overhead is necessary.

\subsubsection{Real-world program patterns}\label{sec:compiler-revisiting-programs}

In general, real-world quantum algorithms can be categorized into two types: (1) those solving classical problems, i.e., quantum versions of digital logics, constructed through binary to qubit encoding, and (2) those solving more quantum or optimization problems, typically constructed through Hamiltonian simulation. For the latter type of programs, there are already highly-effective global compilation techniques that can be seamlessly applied to the $ \SU(4) $ ISA, especially for those ISA-independent dedicated compilers such as Rustiq~\cite{debrugière2024faster}, QuCLEAR~\cite{liu2025quclear}, and PHOENIX~\cite{yang2025phoenix}. Therefore, we focus on how to globally optimize the first type of quantum programs.

Specifically, these programs are composed of $\mathrm{CX}$/$\mathrm{CCX}$/$\mathrm{MCX}$ gates (sometimes with 1Q rotations and Paulis), as exemplified in \Cref{fig:program_example}. This feature is ideal for a \emph{template-based synthesis} strategy---first find the optimal synthesis schemes for these high-level semantics and then unroll them within circuits. We refine 3Q blocks as the IRs and construct their optimal synthesis templates again via approximate synthesis, since intricate components such as the Peres gate~\cite{thapliyal2009design} and MAJ/UMA~\cite{cuccaro2004new} are usually represented as 3Q circuits. Moreover, 3Q represents the perfect granularity for efficient yet powerful synthesis as discussed in \Cref{sec:compiler-hierarchical-synthesis-hyperparameter}. For $\mathrm{MCX}$, it can be first recursively decomposed into $\mathrm{CCX}$ gates~\cite{barenco1995elementary}. This approach simultaneously enables effective optimization and benefits calibration efficiency by reusing a finite set of highly optimized building blocks.

\begin{figure}[tbp]
    \centering
    \includegraphics[width=\columnwidth]{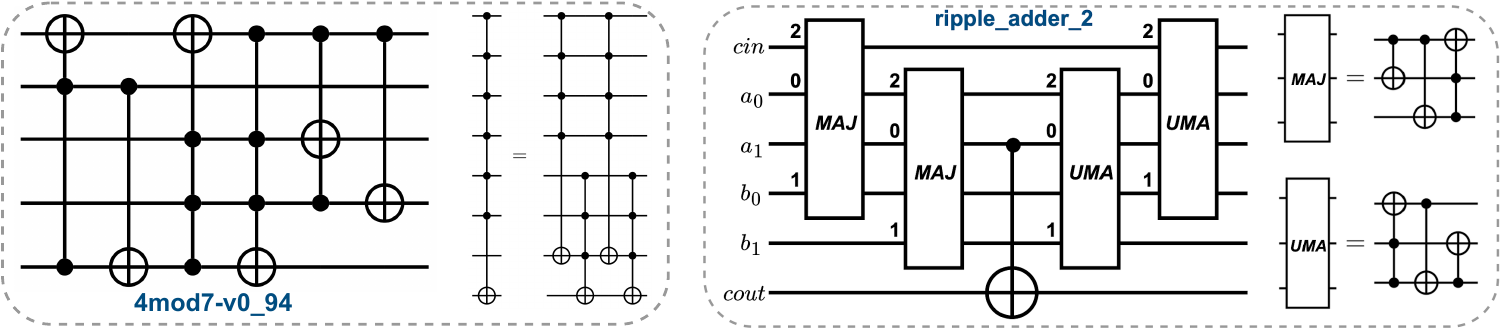}
    \caption{Examples of high-level program IRs.}
    \label{fig:program_example}
\end{figure}

\begin{figure}[tbp]
    \centering
    \includegraphics[width=\columnwidth]{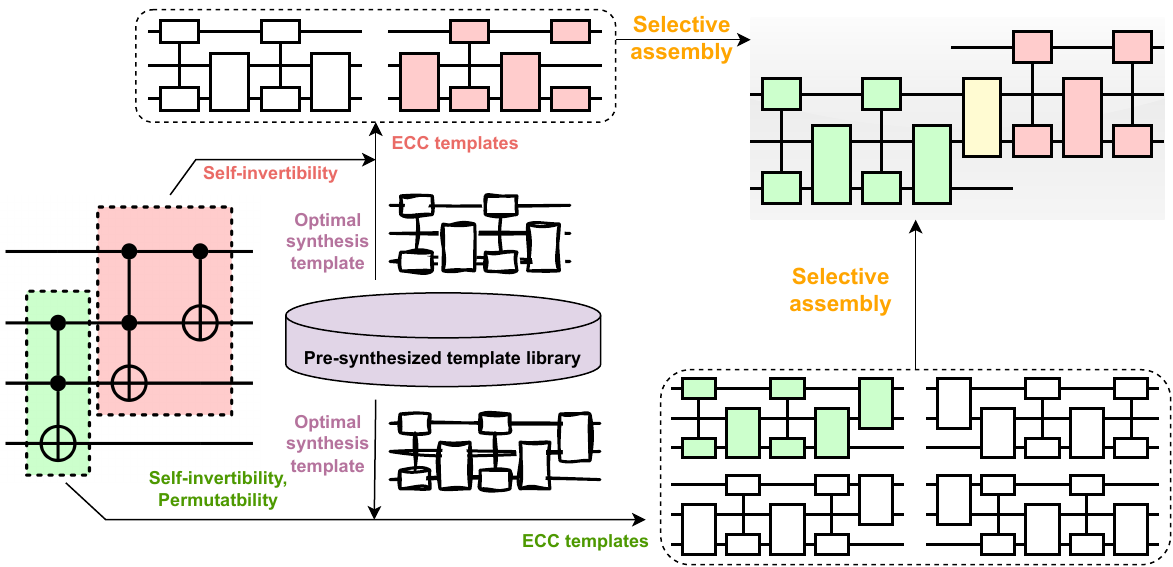}
    \caption{Template-based synthesis building on a pre-synthesized template library and selective assembly.}
    \label{fig:ecc_assembling}
\end{figure}

\subsubsection{Template-based synthesis}
The ReQISC compiler employs a \dquote{pre-synthesis followed by assembly} strategy for compiling programs in their high-level IRs. In the \emph{pre-synthesis} stage, all of the circuits are transformed into $\mathrm{CCX}$-based circuits via $\mathrm{MCX}$ decomposition~\cite{barenco1995elementary}.
The compiler extracts all distinct 3Q IRs from the benchmark suite, ignoring specific qubit indices, and determines their minimal-\#2Q $\SU(4)$ synthesis templates through multiple rounds of approximate synthesis. In the \emph{assembly} stage, the compiler unrolls 3Q IRs in an input $\mathrm{CCX}$-based circuit into their optimal synthesis schemes building on the pre-synthesized template library. The \emph{equivalent circuit classes (ECC)} of each IR and its $\SU(4)$ synthesis templates are derived, according to the \emph{self-invertibility} and \emph{control-bit permutability} of the 3Q IR. This enables further $\SU(4)$ gate count reduction when selectively assembling near-neighbor IRs from their alternative ECC templates, as two $\SU(4)$s acting on the same pair of qubits can be fused into one $\SU(4)$. \Cref{fig:ecc_assembling} illustrates this synthesis procedure of a circuit snippet with consecutive Toffoli and Peres gates. Since the distinct 3Q IRs in real-world programs are finite, this template-based approach is computationally efficient and calibration-friendly.




\subsection{Pushing toward practicality on hardware}\label{sec:compiler-calibration}

\subsubsection{Calibration overhead}\label{sec:compiler-practicality-calibration}
Herein we discuss how the ReQISC compiler overcomes the calibration issue. First, the number of distinct $\SU(4)$s introduced by template-based synthesis is negligible, as the distinct 3Q IR patterns are finite in real-world applications~\cite{wille2008revlib}. 
It is the hierarchical synthesis pass that incurs appreciable calibration overhead through exhaustive local optimization, in which approximate synthesis inevitably induces more diverse $\SU(4)$s' parameters. Therefore, employing the hierarchical synthesis pass means sacrificing calibration efficiency for greater \#2Q reduction.


Moreover, na\"ively applying our scheme to variational quantum programs (e.g., QAOA~\cite{farhi2014quantum}, UCCSD~\cite{barkoutsos2018quantum}) which necessitate multiple runs with variational gate parameters is unrealistic as it requires continual calibration of variational $ \SU(4) $s in experiments. Instead, the workload of reconfiguring variational $\SU(4)$ gates can be shifted to reconfiguring 1Q gates, by decomposing $\SU(4)$s into a gate set with fixed 2Q gates (e.g., $\SQiSW$, $\B$) and parametrized 1Q gates. These 1Q parameters can be easily calibrated with constant experimental overhead by means of the PMW protocol~\cite{chen2023compiling,gong2023robust,wei2024native,chen2025efficient}, which allows for the implementation of any continuous 1Q gate without explicit calibration by simply tuning the phase shift of microwave pulses. While this compromise may result in a slightly higher 2Q gate count, it represents a necessary trade-off for executing variational algorithms.

\subsubsection{Qubit routing} 
Even with sophisticated mapping algorithms, significant \#2Q overhead is induced during qubit routing by the insertion of $\SWAP$ gates~\cite{li2019tackling,liu2022not,zhang2021time}. Leveraging the expressiveness of $\SU(4)$, we develop a dedicated mapping algorithm, \emph{mirroring-SABRE}, by incorporating the $\SU(4)$-aware $\SWAP$ search strategy into the SABRE~\cite{li2019tackling}.

\begin{figure}[tbp]
    \centering
    \includegraphics[width=\columnwidth]{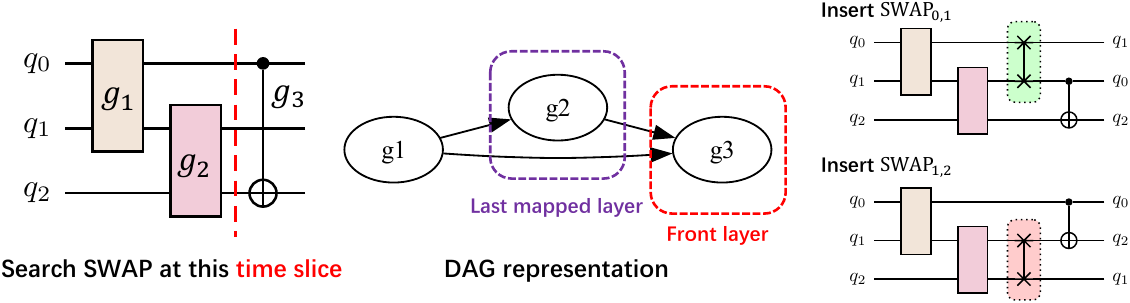}
    \caption{The flexibility of $\SWAP$ insertions is exploited to mitigate \#2Q overhead for qubit mapping on a 1D chain topology ($q_0 \leftrightarrow q_1 \leftrightarrow q_2$). Inserting $\SWAP_{1,2}$ is preferred over $\SWAP_{0,1}$, as it does not increase \#$\SU(4)$.}
    \label{fig:mirroring_sabre}
\end{figure}

Mirroring-SABRE inherits the basic thought of SABRE~\cite{li2019tackling}. SABRE attempts to map 2Q gates layer by layer via extracting the \dquote{front layer} $F$, peeling executable gates and searching $\SWAP$ gates to minimize the heuristic cost $H$ of the unresolved front layer. $H$ involving the current front layer and a $\SWAP$ candidate is designed to minimize the topological distance between upcoming qubits and promote parallelism. In our proposed mirroring-SABRE, we additionally define the \dquote{last mapped layer} $L$ as the set of 2Q gates that has no succeeding ones within the DAG constructed by already mapped 2Q gates, such that the $\SWAP$ search process prioritizes $\SWAP$ gates that $L$ can absorb, as illustrated in \Cref{fig:mirroring_sabre}. We also define an initial heuristic cost function $H_0(F, \mathrm{DAG}, \pi, D)$ before each $\SWAP$ search epoch with a similar calculation method
\begin{align*}
    H_0 := \frac{1}{\left|{F}\right|} \sum\nolimits_{g \in F} D_{\pi[g.q_1], \pi[g.q_2]} + \frac{W}{\left|{E}\right|} \sum\nolimits_{g \in E} D_{\pi[g.{q_1}],\pi[g.{q_2}]},
\end{align*}
relative to the original cost function $H(F, \mathrm{DAG}, \pi, D, \SWAP)$ from SABRE. 
The $\SWAP$ mirroring mechanism prioritizes $ \SWAP $ candidates that can be absorbed by $L$ while simultaneously reducing the heuristic cost ($H < H_0$), thus inducing no \#2Q overhead. If no such candidate meets these criteria, the $\SWAP$ search proceeds according to the $\mathrm{SABRE}$ heuristic.


Besides the flexibility to modify the $\SWAP$ search process, we choose SABRE as the backbone due to its \dquote{lookahead} mechanism. As SABRE evaluates qubit distances in upcoming layers apart from the front layer $ F $, physical qubits involved in the last mapped layer and those within $F$ tend to aggregate after multiple search periods. This emergent aggregation creates more opportunities to fuse a required $\SWAP$ with a preceding 2Q gate, thus reducing routing overhead.



\subsection{Overall framework and implementation}

\subsubsection{End-to-end compilation workflow}\label{sec:compiler-framework-worflow}

The ReQISC compiler pipeline consists of program-aware template-based synthesis, hardware-agnostic hierarchical synthesis, and finally hardware-aware Mirroring-SABRE routing. We perform the hardware-agnostic optimization before routing to maximize gate count reduction free from topological constraints. This approach allows us to approach theoretical \#$\SU(4)$ lower bound or less for each subcircuit. We observe that reversing this order provides no benefit while hindering flexibility and scalability.

\subsubsection{Implementation}

The ReQISC compiler is developed in Python to be a self-contained quantum programming {SDK}. It also operationalizes the ReQISC microarchitecture, enabling the generation of executable $\SU(4)$-based circuits with control parameters applicable to specific hardwares. In case of no available hardware, the output circuits are expressed in $\{ \CanGate{x}{y}{z}, \mathrm{U3}(\theta,\phi,\lambda) \}$. The approximate synthesis functionality leverages modules of BQSKit~\cite{bqskit}. A key feature is its configurable approach to the calibration-performance trade-off. We provide two practical compilation schemes: \dquote{\CompilerFull}, which employs all optimization passes including the hierarchical synthesis for aggressive gate count reduction, and \dquote{\CompilerEff}, which bypasses this pass for minimal calibration overhead.

\begin{table}[t]
    \centering
    \caption{\note{Benchmark suite characteristics. Metrics per category include qubit count range (\#Qubit), 2Q gate count (\numTwoQubit), 2Q depth (\depthTwoQubit), and original circuit duration ($T$). Duration is reported in units of $g^{-1}$, assuming a baseline $ \CNOT $ gate duration of $\pi/\sqrt{2}g$.}}

    \label{tab:benchmarks}
    \begin{footnotesize}
    \begin{tabular}{|l|r|r|r|r|}
        \hline
        \textbf{Category (\#)} & \textbf{\#Qubit} & \textbf{\numTwoQubit} & \textbf{\depthTwoQubit} & \textbf{Duration $T$} \\
        \hline
        alu (12) & 5-6 & 15-192 & 12-166 & 26.7-368.8 \\
        \hline
        bit\_adder (13) & 4-15 & 14-5.1k & 13-3.8k & 28.9-8.5k \\
        \hline
        comparator (19) & 5-6 & 13-155 & 12-134 & 26.7-297.7 \\
        \hline
        encoding (9) & 3-15 & 10-3.3k & 10-2.5k & 22.2-5.6k \\
        \hline
        grover (1) & 9-9 & 288-288 & 228-228 & 506.5-506.5 \\
        \hline
        hwb (12) & 4-170 & 26-4.4k & 25-3.5k & 55.5-7.8k \\
        \hline
        modulo (8) & 5-7 & 9-62 & 8-56 & 17.8-124.4 \\
        \hline
        mult (3) & 12-48 & 99-1.6k & 60-293 & 133.3-650.9 \\
        \hline
        pf (9) & 10-30 & 60-1.8k & 60-1.8k & 133.3-4.0k \\
        \hline
        qaoa (9) & 8-24 & 192-13.2k & 60-1.2k & 133.3-2.7k \\
        \hline
        qft (3) & 8-32 & 56-612 & 26-122 & 57.8-271 \\
        \hline
        ripple\_add (4) & 12-62 & 81-481 & 68-393 & 151.1-873 \\
        \hline
        square (3) & 10-13 & 745-3.8k & 547-3.0k & 1.2k-6.6k \\
        \hline
        sym (6) & 7-14 & 107-29.3k & 63-22.1k & 140-49.1k \\
        \hline
        tof (4) & 5-19 & 18-102 & 15-78 & 33.3-173.3 \\
        \hline
        uccsd (14) & 8-14 & 630-14.4k & 615-13.6k & 1.4k-30.2k \\
        \hline
        urf (3) & 8-9 & 9.0k-23.6k & 7.3k-18.6k & 16.2k-41.3k \\
        \hline
        \textbf{Overall} (132) & 3-170 & 9-3.3k & 8-2.2k & 17.8-49.1k \\
        \hline
    \end{tabular}
    \end{footnotesize}
\end{table}

\section{Evaluation}\label{sec:Evaluation}

\begin{table*}[t]
    \centering
    \caption{\note{Logical-level compilation comparisons in terms of the reduction rates of \#2Q, Depth2Q, and pulse duration. \textsc{Eff.}/\textsc{Full.} denote ReQISC-Eff/Full. For near-identity $\SU(4)$ gates (in \code{qft}, \code{pf}, \code{qaoa} and \code{uccsd}), the gate mirroring technique introduced in \Cref{sec:gate_mirroring} is applied. Durations are evaluated under the $XY$-coupled Hamiltonian, with baseline $\CNOT$ duration $\tau = \pi / \sqrt{2}g$.}}
    
    
    \label{tab:result}
    \begin{footnotesize}

    \begin{tabular}{|l||r|r|r|r|r||r|r|r|r|r||r|r|r|r|r|}
        \hline
        \textbf{Benchmarks}
        & \multicolumn{5}{c||}{\textbf{Average reduction of \numTwoQubit\ (\%)}}
        & \multicolumn{5}{c||}{\textbf{Average reduction of \depthTwoQubit\ (\%)}}
        & \multicolumn{5}{c|}{\textbf{Average reduction of Duration (\%)}} \\
        \hline
        {Category (\#)} & {Qiskit} & {TKet} & {BQSKit} & {\textsc{Eff.}} & {\textsc{Full.}} & {Qiskit} & {TKet} & {BQSKit} & {\textsc{Eff.}} & {\textsc{Full.}} & {Qiskit} & {TKet} & {BQSKit} & {\textsc{Eff.}} & {\textsc{Full.}} \\
        \hline
        alu (12) & 1.54\cellcolor{SkyBlue!15} & 1.72\cellcolor{SkyBlue!30} & 8.12\cellcolor{SkyBlue!45} & 29.83\cellcolor{SkyBlue!60} & 32.01\cellcolor{SkyBlue!70} & 1.85\cellcolor{SkyBlue!15} & 2.20\cellcolor{SkyBlue!30} & 9.20\cellcolor{SkyBlue!45} & 26.41\cellcolor{SkyBlue!60} & 28.40\cellcolor{SkyBlue!70} & 1.85\cellcolor{SkyBlue!15} & 2.20\cellcolor{SkyBlue!30} & 9.20\cellcolor{SkyBlue!45} & 52.41\cellcolor{SkyBlue!60} & 54.48\cellcolor{SkyBlue!70} \\
        \hline
        bit\_adder (13) & 9.43\cellcolor{SkyBlue!15} & 9.70\cellcolor{SkyBlue!30} & 13.81\cellcolor{SkyBlue!45} & 33.01\cellcolor{SkyBlue!60} & 40.38\cellcolor{SkyBlue!70} & 8.21\cellcolor{SkyBlue!15} & 8.36\cellcolor{SkyBlue!30} & 11.33\cellcolor{SkyBlue!45} & 22.75\cellcolor{SkyBlue!60} & 30.45\cellcolor{SkyBlue!70} & 8.21\cellcolor{SkyBlue!15} & 8.36\cellcolor{SkyBlue!30} & 11.33\cellcolor{SkyBlue!45} & 52.57\cellcolor{SkyBlue!60} & 56.14\cellcolor{SkyBlue!70} \\
        \hline
        comparator (19) & 1.13\cellcolor{SkyBlue!15} & 6.66\cellcolor{SkyBlue!45} & 6.15\cellcolor{SkyBlue!30} & 27.56\cellcolor{SkyBlue!60} & 32.50\cellcolor{SkyBlue!70} & 1.29\cellcolor{SkyBlue!15} & 7.11\cellcolor{SkyBlue!45} & 3.86\cellcolor{SkyBlue!30} & 21.91\cellcolor{SkyBlue!60} & 26.25\cellcolor{SkyBlue!70} & 1.29\cellcolor{SkyBlue!15} & 7.11\cellcolor{SkyBlue!45} & 3.86\cellcolor{SkyBlue!30} & 47.17\cellcolor{SkyBlue!60} & 50.44\cellcolor{SkyBlue!70} \\
        \hline
        encoding (9) & 5.66\cellcolor{SkyBlue!15} & 7.78\cellcolor{SkyBlue!30} & 13.20\cellcolor{SkyBlue!45} & 36.49\cellcolor{SkyBlue!60} & 38.88\cellcolor{SkyBlue!70} & 5.66\cellcolor{SkyBlue!15} & 7.86\cellcolor{SkyBlue!30} & 12.40\cellcolor{SkyBlue!45} & 33.94\cellcolor{SkyBlue!60} & 36.46\cellcolor{SkyBlue!70} & 5.66\cellcolor{SkyBlue!15} & 7.86\cellcolor{SkyBlue!30} & 12.40\cellcolor{SkyBlue!45} & 57.15\cellcolor{SkyBlue!60} & 59.14\cellcolor{SkyBlue!70} \\
        \hline
        grover (1) & 0.00\cellcolor{SkyBlue!15} & 0.00\cellcolor{SkyBlue!30} & 17.36\cellcolor{SkyBlue!45} & 44.44\cellcolor{SkyBlue!60} & 44.44\cellcolor{SkyBlue!70} & 0.00\cellcolor{SkyBlue!15} & 0.00\cellcolor{SkyBlue!30} & 20.61\cellcolor{SkyBlue!45} & 30.70\cellcolor{SkyBlue!60} & 30.70\cellcolor{SkyBlue!70} & 0.00\cellcolor{SkyBlue!15} & 0.00\cellcolor{SkyBlue!30} & 20.61\cellcolor{SkyBlue!45} & 53.69\cellcolor{SkyBlue!60} & 53.69\cellcolor{SkyBlue!70} \\
        \hline
        hwb (12) & 1.28\cellcolor{SkyBlue!15} & 3.62\cellcolor{SkyBlue!30} & 9.95\cellcolor{SkyBlue!45} & 25.89\cellcolor{SkyBlue!60} & 31.80\cellcolor{SkyBlue!70} & 1.20\cellcolor{SkyBlue!15} & 4.09\cellcolor{SkyBlue!30} & 10.69\cellcolor{SkyBlue!45} & 27.18\cellcolor{SkyBlue!60} & 34.17\cellcolor{SkyBlue!70} & 1.20\cellcolor{SkyBlue!15} & 4.09\cellcolor{SkyBlue!30} & 10.69\cellcolor{SkyBlue!45} & 51.52\cellcolor{SkyBlue!60} & 58.11\cellcolor{SkyBlue!70} \\
        \hline
        modulo (8) & 9.64\cellcolor{SkyBlue!15} & 9.84\cellcolor{SkyBlue!30} & 14.92\cellcolor{SkyBlue!45} & 34.18\cellcolor{SkyBlue!60} & 38.90\cellcolor{SkyBlue!70} & 10.67\cellcolor{SkyBlue!15} & 10.67\cellcolor{SkyBlue!30} & 15.51\cellcolor{SkyBlue!45} & 32.49\cellcolor{SkyBlue!60} & 38.80\cellcolor{SkyBlue!70} & 10.67\cellcolor{SkyBlue!15} & 10.67\cellcolor{SkyBlue!30} & 15.51\cellcolor{SkyBlue!45} & 58.94\cellcolor{SkyBlue!60} & 62.74\cellcolor{SkyBlue!70} \\
        \hline
        mult (3) & 0.00\cellcolor{SkyBlue!15} & 0.25\cellcolor{SkyBlue!45} & 0.00\cellcolor{SkyBlue!30} & 16.05\cellcolor{SkyBlue!60} & 16.05\cellcolor{SkyBlue!70} & 0.00\cellcolor{SkyBlue!15} & 0.24\cellcolor{SkyBlue!45} & 0.00\cellcolor{SkyBlue!30} & 11.53\cellcolor{SkyBlue!60} & 11.53\cellcolor{SkyBlue!70} & 0.00\cellcolor{SkyBlue!15} & 0.24\cellcolor{SkyBlue!45} & 0.00\cellcolor{SkyBlue!30} & 41.00\cellcolor{SkyBlue!60} & 41.00\cellcolor{SkyBlue!70} \\
        \hline
        pf (9) & 0.96\cellcolor{SkyBlue!15} & 27.41\cellcolor{SkyBlue!45} & 2.20\cellcolor{SkyBlue!30} & 87.59\cellcolor{SkyBlue!60} & 87.59\cellcolor{SkyBlue!70} & 0.96\cellcolor{SkyBlue!15} & 71.85\cellcolor{SkyBlue!45} & 3.46\cellcolor{SkyBlue!30} & 98.63\cellcolor{SkyBlue!60} & 98.63\cellcolor{SkyBlue!70} & 0.96\cellcolor{SkyBlue!15} & 71.85\cellcolor{SkyBlue!45} & 3.46\cellcolor{SkyBlue!30} & 98.70\cellcolor{SkyBlue!60} & 98.70\cellcolor{SkyBlue!70} \\
        \hline
        qaoa (9) & 0.00\cellcolor{SkyBlue!15} & 0.00\cellcolor{SkyBlue!30} & 0.00\cellcolor{SkyBlue!45} & 58.47\cellcolor{SkyBlue!60} & 58.47\cellcolor{SkyBlue!70} & 0.00\cellcolor{SkyBlue!15} & 0.00\cellcolor{SkyBlue!30} & 0.00\cellcolor{SkyBlue!45} & 70.96\cellcolor{SkyBlue!60} & 70.96\cellcolor{SkyBlue!70} & 0.00\cellcolor{SkyBlue!15} & 0.00\cellcolor{SkyBlue!30} & 0.00\cellcolor{SkyBlue!45} & 80.78\cellcolor{SkyBlue!60} & 80.78\cellcolor{SkyBlue!70} \\
        \hline
        qft (3) & 0.00\cellcolor{SkyBlue!15} & 0.00\cellcolor{SkyBlue!30} & 0.00\cellcolor{SkyBlue!45} & 50.00\cellcolor{SkyBlue!60} & 50.00\cellcolor{SkyBlue!70} & 0.00\cellcolor{SkyBlue!15} & 0.00\cellcolor{SkyBlue!30} & 0.00\cellcolor{SkyBlue!45} & 50.00\cellcolor{SkyBlue!60} & 50.00\cellcolor{SkyBlue!70} & 0.00\cellcolor{SkyBlue!15} & 0.00\cellcolor{SkyBlue!30} & 0.00\cellcolor{SkyBlue!45} & 50.77\cellcolor{SkyBlue!60} & 50.77\cellcolor{SkyBlue!70} \\
        \hline
        ripple\_add (4) & 0.00\cellcolor{SkyBlue!15} & 6.21\cellcolor{SkyBlue!30} & 12.37\cellcolor{SkyBlue!45} & 31.06\cellcolor{SkyBlue!60} & 49.57\cellcolor{SkyBlue!70} & 0.00\cellcolor{SkyBlue!15} & 0.72\cellcolor{SkyBlue!30} & 11.86\cellcolor{SkyBlue!45} & 30.82\cellcolor{SkyBlue!60} & 39.07\cellcolor{SkyBlue!70} & 0.00\cellcolor{SkyBlue!15} & 0.72\cellcolor{SkyBlue!30} & 11.86\cellcolor{SkyBlue!45} & 53.75\cellcolor{SkyBlue!60} & 58.69\cellcolor{SkyBlue!70} \\
        \hline
        square (3) & 0.02\cellcolor{SkyBlue!15} & 0.20\cellcolor{SkyBlue!30} & 9.49\cellcolor{SkyBlue!45} & 31.42\cellcolor{SkyBlue!60} & 38.16\cellcolor{SkyBlue!70} & 0.00\cellcolor{SkyBlue!15} & 0.26\cellcolor{SkyBlue!30} & 5.59\cellcolor{SkyBlue!45} & 19.75\cellcolor{SkyBlue!60} & 24.25\cellcolor{SkyBlue!70} & 0.00\cellcolor{SkyBlue!15} & 0.26\cellcolor{SkyBlue!30} & 5.59\cellcolor{SkyBlue!45} & 45.98\cellcolor{SkyBlue!60} & 50.95\cellcolor{SkyBlue!70} \\
        \hline
        sym (6) & 3.42\cellcolor{SkyBlue!15} & 3.86\cellcolor{SkyBlue!30} & 9.67\cellcolor{SkyBlue!45} & 30.16\cellcolor{SkyBlue!60} & 35.42\cellcolor{SkyBlue!70} & 0.62\cellcolor{SkyBlue!15} & 1.31\cellcolor{SkyBlue!30} & 4.67\cellcolor{SkyBlue!45} & 18.93\cellcolor{SkyBlue!60} & 23.81\cellcolor{SkyBlue!70} & 0.62\cellcolor{SkyBlue!15} & 1.31\cellcolor{SkyBlue!30} & 4.67\cellcolor{SkyBlue!45} & 45.92\cellcolor{SkyBlue!60} & 50.75\cellcolor{SkyBlue!70} \\
        \hline
        tof (4) & 0.00\cellcolor{SkyBlue!15} & 4.46\cellcolor{SkyBlue!45} & 0.00\cellcolor{SkyBlue!30} & 30.29\cellcolor{SkyBlue!60} & 30.29\cellcolor{SkyBlue!70} & 0.00\cellcolor{SkyBlue!15} & 3.51\cellcolor{SkyBlue!45} & 0.00\cellcolor{SkyBlue!30} & 12.38\cellcolor{SkyBlue!60} & 12.38\cellcolor{SkyBlue!70} & 0.00\cellcolor{SkyBlue!15} & 3.51\cellcolor{SkyBlue!45} & 0.00\cellcolor{SkyBlue!30} & 39.39\cellcolor{SkyBlue!60} & 39.39\cellcolor{SkyBlue!70} \\
        \hline
        uccsd (14) & 23.46\cellcolor{SkyBlue!30} & 64.03\cellcolor{SkyBlue!45} & 6.71\cellcolor{SkyBlue!15} & 76.50\cellcolor{SkyBlue!60} & 84.11\cellcolor{SkyBlue!70} & 23.49\cellcolor{SkyBlue!30} & 67.53\cellcolor{SkyBlue!45} & 7.48\cellcolor{SkyBlue!15} & 78.45\cellcolor{SkyBlue!60} & 85.63\cellcolor{SkyBlue!70} & 23.49\cellcolor{SkyBlue!30} & 67.53\cellcolor{SkyBlue!45} & 7.48\cellcolor{SkyBlue!15} & 84.02\cellcolor{SkyBlue!60} & 89.25\cellcolor{SkyBlue!70} \\
        \hline
        urf (3) & 0.91\cellcolor{SkyBlue!15} & 4.47\cellcolor{SkyBlue!30} & 6.81\cellcolor{SkyBlue!45} & 20.61\cellcolor{SkyBlue!60} & 25.90\cellcolor{SkyBlue!70} & 0.83\cellcolor{SkyBlue!15} & 4.66\cellcolor{SkyBlue!45} & 3.30\cellcolor{SkyBlue!30} & 14.34\cellcolor{SkyBlue!60} & 20.77\cellcolor{SkyBlue!70} & 0.83\cellcolor{SkyBlue!15} & 4.66\cellcolor{SkyBlue!45} & 3.30\cellcolor{SkyBlue!30} & 42.71\cellcolor{SkyBlue!60} & 48.55\cellcolor{SkyBlue!70} \\
        \hline
        \textbf{Overall} (132) & \emph{5.34}\cellcolor{SkyBlue!15} & \emph{15.91}\cellcolor{SkyBlue!45} & \emph{7.99}\cellcolor{SkyBlue!30} & \emph{46.95}\cellcolor{SkyBlue!60} & \emph{51.89}\cellcolor{SkyBlue!70} 
            & \emph{5.2}\cellcolor{SkyBlue!15} & \emph{21.83}\cellcolor{SkyBlue!30} & \emph{7.34}\cellcolor{SkyBlue!45} & \emph{53.43}\cellcolor{SkyBlue!60} & \emph{57.5}\cellcolor{SkyBlue!70} 
            & \emph{5.2}\cellcolor{SkyBlue!15} & \emph{21.83}\cellcolor{SkyBlue!30} & \emph{7.34}\cellcolor{SkyBlue!45} & \emph{68.03}\cellcolor{SkyBlue!60} & \emph{71.0}\cellcolor{SkyBlue!70} \\
        \hline
    \end{tabular}
    \end{footnotesize}

\end{table*}


We evaluate ReQISC against state-of-the-art (SOTA) counterparts on a comprehensive benchmark suite, demonstrating its superiority across multiple hardware types and showing clear advantages in its microarchitecture, as well as in end-to-end logical-level (topology-agnostic) and topology-aware compilation. All experiments were run on an Apple M3 Max laptop.




\subsection{Experimental setup}

\subsubsection{Metrics}
We focus on \numTwoQubit\ (2Q gate count), \depthTwoQubit\ (depth of the circuit involving only 2Q gates), {pulse duration, and program fidelity} as metrics.
The whole circuit duration corresponds to the critical path with the longest pulse duration. Program fidelity is measured by the Hellinger fidelity between noisy simulation and ideal simulation.


\subsubsection{Baselines}\label{sec:evaluation-baseline}
\dquote{Qiskit}, \dquote{TKet}, and \dquote{BQSKit} are three primary baselines for comparison. In Qiskit, the standard O3 \code{transpile} pass is used. In Tket, the standard compilation procedure is constructed by sequentially employing \code{PauliSimp} and \code{FullPeepholeOptimise} passes, {in which \code{PauliSimp} is the built-in Hamiltonian simulation program optimization engine~\cite{cowtan2019phase}}. To demonstrate that ReQISC's benefits stem from its co-designed compilation strategy, we also compare against {three} variants: \dquote{Qiskit-$\SU(4)$} and \dquote{TKet-$\SU(4)$}, with a pass that partitions and fuses 2Q gates into $ \SU(4) $s appended to the standard Qiskit/TKet compilation procedure, and \dquote{BQSKit-$\SU(4)$}, which employs its \code{compile} pass by designating a customized $\{\Can, \mathrm{U3}\}$ gate set.




\subsubsection{Benchmarks}\label{sec:benchmarks}
The benchmark suite comprises $ 132 $ programs from $ 17 $ categories (\code{alu}, \code{encoding}, etc.) shown in \Cref{tab:benchmarks}. Program sizes indicated by \#2Q in $\CNOT$-based circuit representation range from tens to $10^4$. Most programs are selected from published benchmark suites~\cite{wille2008revlib,sivarajah2020t}.

Notations of most benchmarks align with those used in published benchmark suites for which we refer readers to \href{https://www.revlib.org/}{RevLib} and \href{https://github.com/CQCL/tket_benchmarking/}{TKet's benchmarking repo}. For benchmarks whose high-level algorithmic description is provided in \code{.real} format, featuring $\mathrm{MCX}$ gates as basic subroutines, we convert them into $\mathrm{CCX}$-based circuits as the input for the ReQISC compiler and further decompose them into $\CNOT$-based circuits as the input for baselines. For variational quantum programs, Pauli strings with coefficients are compiled into $ \SU(4) $ gate sequences by the high-level ISA-independent PHOENIX compiler~\cite{yang2025phoenix} before being used as the input for ReQISC, while $ \CNOT $-based circuits are fed to baselines.


\begin{table}[tbp]
    \centering
    \caption{Synthesis cost comparison in terms of gate duration $\tau$ ($g^{-1}$) achieved by ReQISC. Haar-random average costs (number of 2Q basis gates required for arbitrary $\SU(4)$ synthesis) for $ \CNOT $, $ \iSWAP $, $ \SQiSW $, $ \B $ gates are 3, 3, 2.21, and 2, respectively.}
    
    \label{tab:gate-duration}
    \setlength{\tabcolsep}{4pt}

    \begin{footnotesize}
        \begin{tabular}{|l|c|c|c|c|c|c|}
            \hline
            \textbf{Coupling} & \multicolumn{2}{c|}{\textbf{$ \mathrm{XY} $ coupling}} & \multicolumn{2}{c|}{\textbf{$ \mathrm{XX} $ coupling}} & \multicolumn{2}{c|}{\textbf{Random}} \\
            \hline
            \textbf{Basis gate} & $\tau$ (Sgl.) & $\tau$ (Avg.) & $\tau$ (Sgl.) & $\tau$ (Avg.) & $\tau$ (Sgl.) & $\tau$ (Avg.) \\
            \hline
            CNOT & 2.221 & 6.664 & -- & -- & -- & -- \\
            \hline
            SU(4) & -- & \textbf{1.341} & -- & \textbf{1.178} & -- & \textbf{1.321} \\
            \hline
            CNOT & 1.571 & 4.712 & 0.785 & 2.356 & 1.228 & 3.684 \\
            \hline
            iSWAP & 1.571 & 4.712 & 1.571 & 4.712 & 1.898 & 5.693 \\
            \hline
            SQiSW & 0.785 & 1.736 & 0.785 & 1.736 & 0.949 & 2.097 \\
            \hline
            B & 1.571 & 4.712 & 1.178 & 2.356 & 1.435 & 2.869 \\
            \hline
        \end{tabular}
    \end{footnotesize}

\end{table}

\subsubsection{Hardware assumptions}
The ReQISC microarchitecture is evaluated under representative coupling Hamiltonians, including $ \mathrm{XY} $, $ \mathrm{XX} $, and random couplings. Since $\mathrm{XY}$ coupling aligns with the mainstream flux-tunable superconducting platforms (e.g., Google's Sycamore~\cite{arute2019quantum}), it is the default assumption for pulse-level benchmarking and fidelity experiments. Baselines are compared using their conventional, optimized pulse implementations for $\CNOT$, $\iSWAP$, and $\SQiSW$ gates under $ \mathrm{XY} $ coupling~\cite{krantz2019quantum,arute2019quantum,huang2023quantum}.


\subsection{Performance of ReQISC microarchitecture}\label{sec:eval-microarch}



We benchmark the ReQISC microarchitecture by computing the average pulse duration to synthesize $10^5$ Haar-random $\SU(4)$s (\Cref{tab:gate-duration}). For mainstream $ \mathrm{XY} $-coupled platforms, ReQISC slashes the average synthesis cost to just $1.341\,g^{-1}$---a nearly five-fold improvement over $6.664\,g^{-1}$ ($3\times \pi/\sqrt{2}\approx 6.664  $~\cite{krantz2019quantum}) required by conventional $\CNOT$ realization. This efficacy is consistent across hardware interactions, with average durations of $1.178\,g^{-1}$ for $ \mathrm{XX} $ coupling and $1.321\,g^{-1}$ for arbitrary random couplings. These results demonstrate that a full $\SU(4)$ ISA, when powered by our microarchitecture, fundamentally outperforms conventional approaches. Our scheme is so efficient that even when used to realize a fixed gate like $\SQiSW$, its performance is nearly on par (within a factor of $1.6$) with that of the fully arbitrary $\SU(4)$.




\subsection{Logical-level compilation}





As detailed in \Cref{tab:result}, ReQISC delivers substantial improvements over all baselines across every key metric. Particularly, ReQISC consistently slashes the total pulse-level circuit duration by $40$\%--$90$\% across all benchmarks, with an overall average of $ 68\% $ ($ 71\% $) for \CompilerEff\ (\CompilerFull). The difference between the results for \CompilerEff\ versus \CompilerFull\ is much more pronounced for non-variational, arithmetic-logic programs, suggesting that they benefit more from local optimization via hierarchical synthesis. That aligns with our expectation since they exhibit versatile subcircuit patterns.


\begin{figure}[tbp]
    \centering

    \includegraphics[width=\columnwidth]{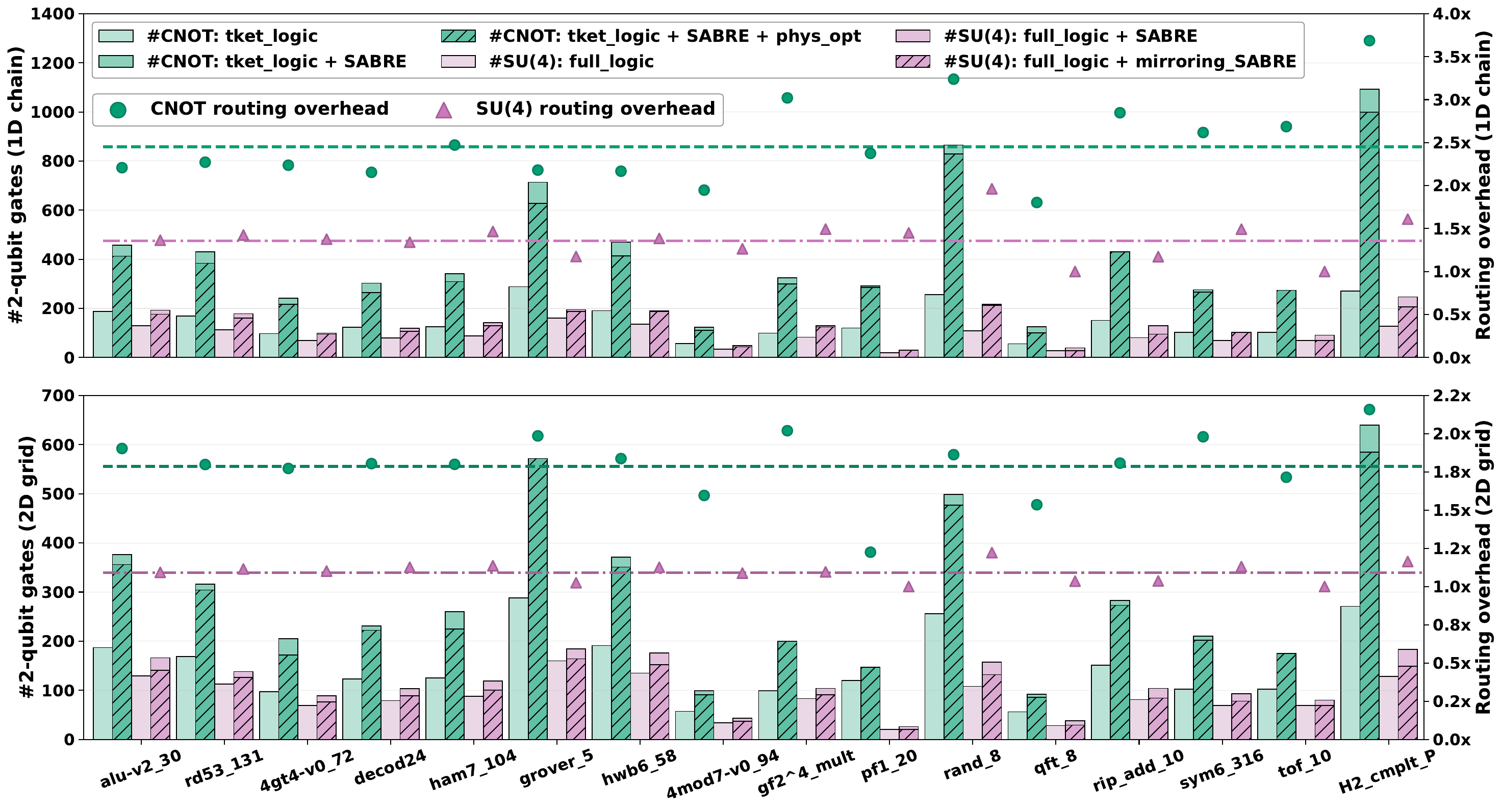}

    \caption{topology-aware benchmarking for 1D-chain (upper) and 2D-grid (bottom) mapping. \code{xxx\_logic}: logical circuits (pre-mapping); \code{phys\_opt}: physical circuits with Qiskit O3 topology-preserving optimization. Green dots and blue triangles indicate the multiples of \#2Q after optimized qubit mapping (hatched-green or hatched-blue) relative to logical circuits (light-green or light-blue). Dashed horizontal lines refer to the geometric-mean overheads.}
    \label{fig:topology-aware}
\end{figure}

\begin{figure}[tbp]
    \centering
    \includegraphics[width=\columnwidth]{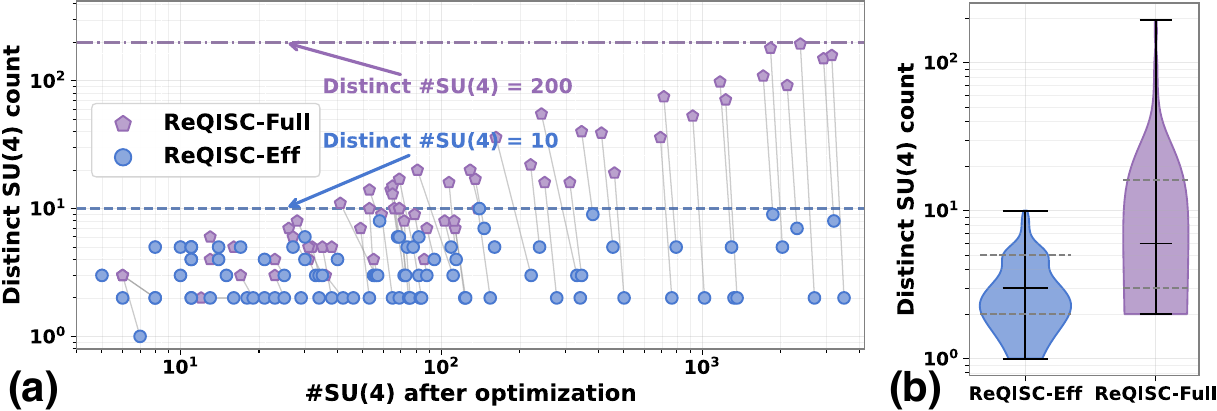}
    \caption{\note{Calibration efficiency and trade-off with \#2Q reduction (\CompilerEff\ vs. \CompilerFull). (a) Each connected pair of points corresponds to a benchmark program. Scaling trends remain consistent for larger circuits (not shown due to fewer data points). (b) Distinct $ \SU(4) $ count distributions.}}
    \label{fig:tradeoff}
\end{figure}

\subsection{Topology-aware compilation}

\Cref{fig:topology-aware} illustrates ReQISC's performance with qubit mapping/routing, across medium-scale benchmarks per category. Results show that mirroring-SABRE outperforms conventional SABRE paired with ReQISC, with \#2Q reduction {up to $28.2$\% ($23.7$\%) and on average $11.0$\% ($15.7$\%)} for 1D-chain (2D-grid) mapping. This effect is also program-dependent---for instance, mirroring-SABRE incurs no additional \#2Q overhead on \code{qft\_8} and \code{tof\_10}. The advantage is even more pronounced when comparing the end-to-end ReQISC compilation against the SOTA $\CNOT$-based baseline (hatched-blue vs. hatched-green). Specifically, the geometric mean of \#$\SU(4)$ overhead is {$1.36$x ($1.09$x)} for 1D-chain (2D-grid) mapping; for \#$\CNOT$, it is $2.45$x ($1.79$x). These results underscore the ReQISC's superiority in topology-aware compilation.

\subsection{Calibration efficiency and tradeoff discussion}\label{sec:tradeoff}




We evaluate the trade-off between calibration overhead and 2Q gate count reduction by comparing the number of distinct $\SU(4)$ gates in circuits compiled by \CompilerEff\ and \CompilerFull\ for benchmarks with \#2Q up to $ 5000 $. As shown in \Cref{fig:tradeoff}, \CompilerEff\ has negligible calibration overhead as expected, with less than 10 distinct $\SU(4)$s. For \CompilerFull, the number of distinct SU(4)s is below $200$ (\Cref{fig:tradeoff}(a)), \note{while more than three quarters of \CompilerFull-compiled programs exhibits less than $ 20 $ distinct 2Q gates (\Cref{fig:tradeoff}(b))}. This demonstrates a practical and controllable trade-off for users seeking the highest performance, allowing for a balance between gate fidelity with experimental complexity.




\note{
\subsubsection{System-level calibration scalability}\label{sec:tradeoff:calibration_scalability}

A critical concern for adopting continuous ISAs is whether the calibration cost scales prohibitively with the number of user programs.
We argue that ReQISC avoids this potential \dquote{calibration explosion} through several synergistic mechanisms.

First, the structural regularity of quantum algorithms ensures a \dquote{bounded template library}. Implementing arbitrary 2Q gates does not require supporting every possible unitary in practice; instead, most programs require only a few distinct $\SU(4)$ operations. 
This implies that a system provider only needs to maintain a calibrated library of a few dozen \dquote{standard gates} to serve a vast range of applications, keeping the marginal cost of supporting new programs negligible.

Second, our framework utilizes \dquote{model-based parameter generation}. Since the ReQISC microarchitecture is physically derived from the coupling Hamiltonian, once the underlying hardware model (e.g., $g$ and drive strengths) is characterized---a standard periodic maintenance task---control parameters for any new $\SU(4)$ can be analytically generated or interpolated with high initial fidelity~\cite{chen2024one,chen2025efficient}. Recent experimental work by \citet{kakkar2025noneed} further validates this approach, demonstrating that a continuous gate family can be implemented with minimal to zero extra calibration cost by characterizing only a small subset of the parameter space.

Furthermore, we emphasize that ReQISC's \dquote{reconfigurability} aligns with existing experimental standards. High-fidelity computation typically necessitates recalibration before running any nontrivial program, even on fixed-ISA systems. The perception of high cost for rich ISAs is largely a byproduct of them not yet being mainstream. In practice, the total calibration cost scales linearly with distinct $\SU(4)$ count. It is experimentally feasible, as evidenced by the successful calibration of six distinct gates in \citet{chen2025efficient}.


While the number of distinct gates affects calibration overhead, a larger gate set may also introduce more complex parallel crosstalk calibration. Fortunately, our unified control framework for all $\SU(4)$ operations significantly simplifies this process. This is demonstrated by the successful calibration of six distinct gates in \citet{chen2025efficient}, as previously noted, and could be further streamlined by the techniques proposed in \citet{kakkar2025noneed}.


}

\subsection{Ablation study and Breakdown analysis}\label{sec:ablation-study}

\begin{figure}[tbp]
    \centering
    \includegraphics[width=\columnwidth]{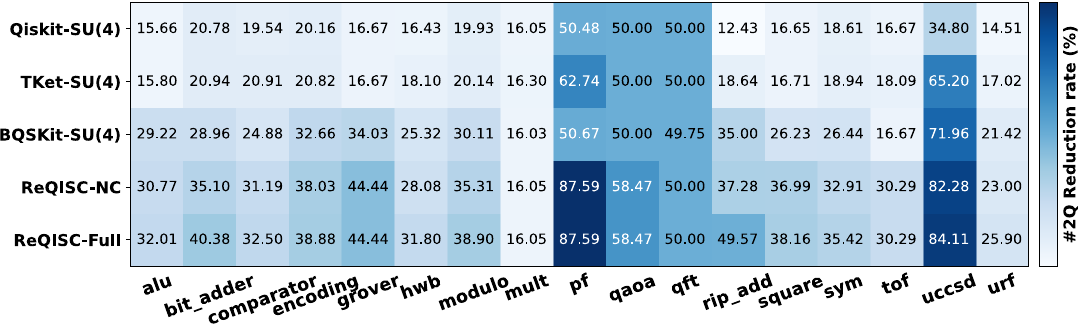}
    \caption{Ablation study benchmarking results.} 
    \label{fig:ablation_study}
\end{figure}

\subsubsection{Efficacy of specialized passes}

To isolate the benefits of our co-designed compilation strategy, we conduct the ablation study for \CompilerFull\ against Qiskit-$\SU(4)$, TKet-$\SU(4)$, and BQSKit-$\SU(4)$. As shown in \Cref{fig:ablation_study}, \CompilerFull\ still significantly outperforms these variants, confirming that its global program-aware and hierarchical synthesis passes are critical for effective $\SU(4)$ compilation. Notably, while BQSKit-$\SU(4)$ achieves a reasonable gate count reduction, it is impractical due to an explosion of distinct $\SU(4)$ gates, which implies an uncontrollable calibration overhead. In contrast, \CompilerFull\ maintains a manageable, sub-linearly scaling set of distinct gates---an order of magnitude fewer in our tests (\Cref{fig:tradeoff}).

\subsubsection{Impact of DAG Compacting}
Our breakdown analysis confirms the importance of the DAG compacting pass. We first observed that different circuit partitioning strategies had a negligible impact on performance (less than $2\%$ variation). The DAG compacting pass, however, is crucial. Comparing \CompilerFull\ against a version without this pass (\dquote{ReQISC-NC}) (\dquote{NC} denotes no-compacting) reveals substantial gains. As shown in \Cref{fig:ablation_study}, this single pass improves the \numTwoQubit\ reduction rate by up to $33\%$ (e.g., for \texttt{rip\_add} benchmarks), demonstrating its effectiveness in exploiting approximate commutation rules.


\subsection{Fidelity experiment}\label{sec:fidelity-experiment}

\begin{figure}[tbp]
    \centering
    \includegraphics[width=\columnwidth]{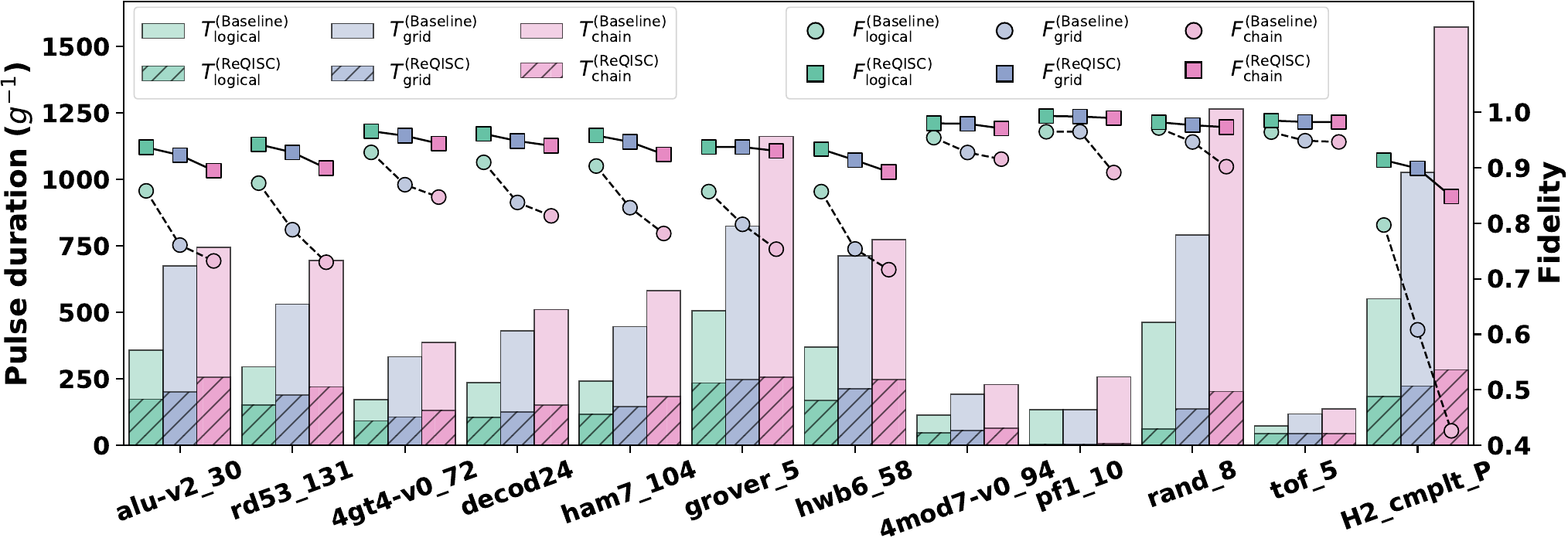}
    \caption{Program fidelity $F$ and pulse duration $T$ comparison through noisy simulation.}
    \label{fig:fidelity-experiment}
\end{figure}

We performed noisy simulations using QSim~\cite{quantum_ai_team_and_collaborators_2020_4023103} on twelve representative benchmarks, comparing ReQISC against a SOTA $\CNOT$-based workflow (TKet with SABRE mapping). Our model assumed a $ \mathrm{XY} $-coupled device with the depolarizing channel appended to 2Q gates, where error rates were scaled proportionally to the gate duration ($\tau$) of each $ \CNOT $ ($ \tau_0 = \pi/\sqrt{2}\,g^{-1};\, p_0 = 0.001 $) or $\SU(4)$ operation ($\tau;\, p=p_0 \tau/\tau_0 $). As shown in \Cref{fig:fidelity-experiment}, ReQISC delivers substantially higher fidelity and faster execution across the board. At the logical level (all-to-all topology), it achieves an average $2.36$x error reduction and a $3.06$x execution speedup. These advantages become even more pronounced when mapping to topology-constrained devices, with error reductions of $3.18$x (2D grid) and $3.34$x (1D chain), and corresponding speedups of $4.30$x and $4.55$x, respectively.

\subsection{Reliability and scalability}\label{sec:reliability-and-scalability}

\begin{figure}[tbp]
    \centering

    \includegraphics[width=\columnwidth]{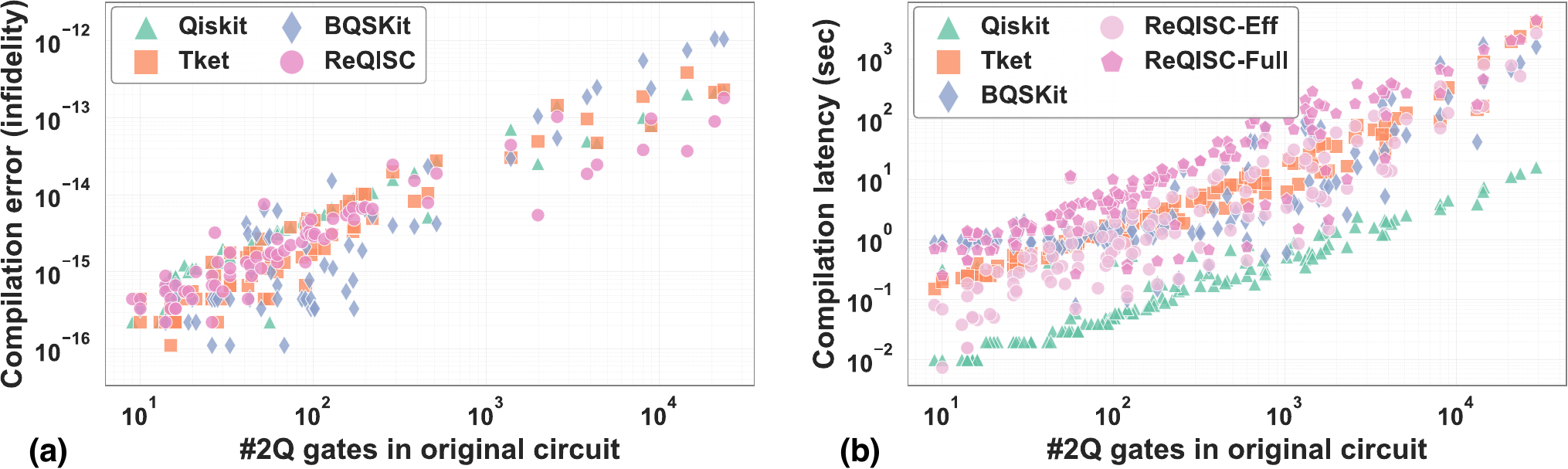}
    
    \caption{{Compilation errors (a) and latency (b).}} 
    \label{fig:reliability-and-scalability}
\end{figure}


For benchmarks comprising up to eleven qubits, we verify the circuits generated by baselines and the ReQISC compiler by computing the compilation error quantified by circuit infidelity. \Cref{fig:reliability-and-scalability}(a) shows that all compilers exhibit a similar level of compilation errors. This is expected, as the approximate synthesis is de facto \dquote{numerically exact} for every subcircuit synthesis, whose correctness is bounded by the IEEE754 machine precision $\approx 2.2\times 10^{-16}$. The extensive peephole optimizations (e.g., block consolidation, KAK decomposition) within Qiskit/TKet can only be at the level of machine precision as well. Note that both $x$ and $y$ axes use logarithmic scaling---for large circuits, ReQISC/Qiskit/TKet's compilation errors are bounded much better than BQSKit.

The ReQISC compiler's computational complexity is polynomial in program size, as its key passes have linear or polynomial complexity. \Cref{fig:reliability-and-scalability}(b) confirms that both \CompilerEff\ and \CompilerFull\ scale well in their runtime, despite being implemented in Python. \CompilerEff\ is consistently more efficient than TKet (C++ backend) and BQSKit (Rust backend). \CompilerFull, which performs heavy optimizations, offers performance competitive with BQSKit, demonstrating that its substantial gate count reduction is achieved with a practical and acceptable increase in compilation time.





\section{Conclusion}\label{sec:Conclusion}

Traditional quantum hardware relies predominantly on $\CNOT$-based ISAs. Although some continuous gate families have been proposed and experimentally demonstrated, they are generally impractical due to suboptimal or non-native implementation, excessive calibration overhead, and limited compilation strategies. ReQISC addresses these limitations with a hardware-native $\SU(4)$ ISA, optimally realized by the underlying microarchitecture across various hardware. For the first time, our approach enables the straightforward implementation of any equivalent 2Q gate under arbitrary coupling Hamiltonians, with simple pulse control and minimal gate time. Supported by the end-to-end ReQISC compiler, we can strike an excellent balance between performance and calibration overhead, thereby challenging the traditional notion that performance and calibration overhead are fundamentally in contradiction.

The implications of this work are significant for quantum computing both in the near and long term. For the fault-tolerant regime~\cite{preskill2025beyond,acharya2024quantum}, our approach offers multifaceted advantages: it provides a pathway to faster and higher-fidelity logical qubits (e.g., via a $\sqrt{2}$x faster $\CNOT$ on XY-coupled hardware) and enables native support for Clifford operations like $\iSWAP$ and $ \SWAP $  that proves critical for SOTA QEC schemes~\cite{mcewen2023relaxing,eickbusch2024demonstrating,zhou2024halma,wu2022synthesis}. By providing the architecture and compiler to leverage a richer, more expressive native gate set, ReQISC opens new avenues for designing performant and scalable quantum computers.

\begin{acks}
  We thank all anonymous reviewers for their thoughtful comments.
  This research was partially conducted by AI Chip Center for Emerging Smart Systems (ACCESS), supported by the InnoHK initiative of the Innovation and Technology Commission of the Hong Kong Special Administrative Region Government. It was also supported partially by Research Grants Council of Hong Kong SAR (\#16213824 \& \#16212825).
  Portions of this work were conducted while Z. Y. was attending UC Santa Barbara and visiting Tsinghua University, and while D. D. was at Tsinghua University.
  D. D. would like to thank God for all of His provisions.
\end{acks}

\appendix
\section*{Artifact Appendix}\label{appendix:artifact_evaluation}

\subsection*{Abstract}

ReQISC couples the Regulus SU(4)-native compiler with the genAshN microarchitecture (time-optimal SU(4) gate scheme). The artifact bundles: (1) the compiler implementation in \texttt{regulus/}; (2) the genAshN gate-duration model in \texttt{microarch/genAshN}; (3) the benchmark suite in \texttt{benchmarks/}; and (4) automation in \texttt{artifact/} (Makefile + scripts + notebooks) to regenerate the paper's tables and figures. The Makefile executes the full compilation pipeline, aggregates metrics, and emits LaTeX tables and PDF figures from the produced CSVs. 

\subsection*{Artifact check-list (meta-information)}
{\small
\begin{itemize}
  \item \textbf{Algorithm:} SU(4)-native template-based synthesis, hierarchical approximate synthesis, mirroring-SABRE qubit mapping; genAshN gate scheme solving
  \item \textbf{Program:} Python 3.8+ (Regulus compiler and scripts)
  \item \textbf{Transformations:} template-based synthesis, hierarchical SU(4) synthesis, mirroring, decomposition, circuit concatenation, circuit-DAG conversion
  \item \textbf{Data set:} \texttt{benchmarks/Type-I}, \texttt{benchmarks/Type-II} QASM/JSON quantum programs with high-level subroutines
  \item \textbf{Run-time environment:} Linux/macOS, Python virtualenv
  \item \textbf{Hardware:} x86\_64 with $>$16\,GB RAM recommended; CPU-bound
  \item \textbf{Run-time state:} deterministic; no RNG required
  \item \textbf{Execution:} Makefile targets wrapping Python scripts
  \item \textbf{Metrics:} two-qubit gate count, two-qubit depth, pulse duration
  \item \textbf{Output:} CSVs in \texttt{artifact/results}, QASM in \texttt{artifact/output}, LaTeX in \texttt{artifact/notebooks/eval\_compiler\_main}
  \item \textbf{Experiments:} compiler benchmarking vs. Qiskit, TKET, BQSKit; pulse-duration comparison
  \item \textbf{Prep time:} $<$10 minutes for environment setup
  \item \textbf{Experiment time:} Demo $<1$ minutes; Full \texttt{make results} several hours (Regulus-Eff/Full + Baselines (Qiskit/TKet/BQSKit))
  \item \textbf{Publicly available?:} Yes
  \item \textbf{Code licenses?:} Apache 2.0
  \item \textbf{Data licenses?:} Same as code (benchmarks included)
  \item \textbf{Workflow automation?:} GNU Make
  \item \textbf{Archived?:} \url{https://zenodo.org/records/18163249} with DOI 10.5281/zenodo.18163249
\end{itemize}
}

\subsection*{Description}
\subsubsection*{How to access}
Artifact is accessible via the Zenodo archive \url{https://zenodo.org/records/18163249}. 

Main entry points:
\begin{itemize}
  \item \texttt{regulus/}: Regulus compiler implementation.
  \item \texttt{microarch/genAshN/}: Microarchitecture implementation (gate scheme)
  \item \texttt{benchmarks/}: Benchmark circuits with metadata.
  \item \texttt{artifact/scripts/}: Command-line interface scripts.
  \item \texttt{artifact/notebooks/}: Jupyter notebooks for figures and tables.
  \item \texttt{artifact/Makefile}: Automation of experiment workflow.
\end{itemize}

\subsubsection*{Hardware dependencies}
Any modern x86\_64 CPU or Apple Silicon; 8 cores and $>$16\,GB RAM recommended. No GPU/accelerator required.

\subsubsection*{Software dependencies}
Python 3.8+, pip, virtualenv. Packages listed in \texttt{requirements.txt} (qiskit, bqskit, tket, rustworkx, tqdm, etc.). Dockerfile provided for containerized runs.

\subsubsection*{Data sets}
Benchmarks, with metadata, data preparation and pre-processing scripts in \texttt{benchmarks/}. No external downloads.

\subsubsection*{Models}
The ReQISC microarchitecture design protocols (\Cref{algo:isa}) implemented in \texttt{microarch/genAshN} provides duration computation model for pulse-level comparisons; no training required.

\subsection*{Installation}

\begin{enumerate}
  \item \textbf{Create env:}
  {\small
    \begin{verbatim}
    cd artifact
    make env
    source .venv/bin/activate
    \end{verbatim}
  }
  \item \textbf{(Optional) Docker:} \texttt{docker build -t regulus .} then run a container and execute the same Make targets inside.
\end{enumerate}

\subsection*{Experiment workflow}

\begin{enumerate}
  \item \textbf{Sanity run (minutes):} \texttt{make demo} (Regulus reduced + Qiskit on category \texttt{alu}).
  \item \textbf{Full regeneration (hours):} \texttt{make results}. Runs Regulus (reduced and full), Qiskit, TKET, and BQSKit on all categories; writes
  {\small
    \begin{itemize}
        \item QASM to \texttt{artifact/output/<compiler>/<category>/}
        \item result to \texttt{artifact/results/<compiler>/<category>/}
        \item aggregated to \texttt{artifact/results/result\_<compiler>.csv}
    \end{itemize}
  }
  \item \textbf{Tables only:} \texttt{make tables} regenerates from the stored CSVs to \texttt{compilation\_result.tex}.
  \item \textbf{Figures:} Run the scripts in \texttt{artifact/notebooks/} (each folder maps to a figure suite).
\end{enumerate}

\subsection*{Evaluation and expected results}

\begin{itemize}
  \item \textbf{Circuit-level metrics:} Two-qubit gate/depth reductions, contained in \texttt{artifact/results/result\_*.csv} contain two-qubit gate/depth reductions matching the paper. Success: regenerated, non-empty CSVs.
  \item \textbf{Pulse metrics:} \texttt{pulse\_comparison.py} under \texttt{artifact/} summarizes duration reductions; \texttt{pulse\_comparison.tex} formats the table.
  \item \textbf{LaTeX tables:} \texttt{compilation\_result.tex} matches manuscript tables.
\end{itemize}

\subsection*{Experiment customization}

\begin{itemize}
  \item Limit categories: edit \texttt{CATEGORIES} in \texttt{artifact/Makefile} or invoke \texttt{scripts/bench.py} on a specific directory.
  \item Runtime vs. quality: set \texttt{TRIALS=<int>} when invoking \texttt{make} to control Regulus hierarchical synthesis trials.
  \item Topology-aware runs: pass \texttt{--device chain} or \texttt{--device grid} to \texttt{bench.py} if desired.
\end{itemize}

\subsection*{Notes}

\begin{itemize}
  \item If TKET or BQSKit are unavailable on your platform, regenerate tables from provided CSVs via \texttt{make tables}.
  \item Delete \texttt{artifact/output/} between runs by executing \texttt{make clean-output} to avoid mixing fresh and archived QASM.
\end{itemize}

\bibliographystyle{ACM-Reference-Format}
\balance
\bibliography{refs}

\appendix
\section{Proof of the \Cref{algo:isa} (genAshN gate scheme)}\label{appendix:proof}

\begin{takeaways}[title={Takeaways}] 
    Since the presented scheme could also be regarded as a \dquote{generalized} version of the AshN scheme and their proof and solving thoughts are sort of similar, we also call the ReQISC microarchitecture as \dquote{genAshN} scheme throughout this appendix. Source code implementation, test files, and usage examples are available in the Zenodo archive \url{https://zenodo.org/records/18163249}.
    

\end{takeaways}


\subsection{Problem Statement \& Canonicalization}
\label{sec:ps}
\subsubsection{Hamiltonian Canonicalization}
Define
\begin{align}
H[a,b,c]:=a\cdot \sigma_X\otimes \sigma_X + b \cdot \sigma_Y\otimes \sigma_Y + c\cdot \sigma_Z \otimes \sigma_Z.
\end{align}
We follow the canonicalization treatment in ~\cite{dur2001entanglement} and only consider two-qubit Hamiltonians of the form $H[a,b,c]$ where $a\geq b\geq |c|,\, a>0$. Furthermore, we can rescale any Hamiltonian as long as the corresponding gate time is rescaled accordingly. We thus identify canonical Hamiltonian coefficients $H[a,b,c]\sim H[ka, kb,kc]$ for all $k>0$.

\subsubsection{Two-qubit gate Canonicalization}
Through the well-known canonical decomposition (or KAK decomposition)~\cite{zhang2003geometric}, any two-qubit unitary can be determined by a unique Weyl coordinate $ (x,y,z) $ such that
\begin{align}
    U &= g \cdot (V_1 \otimes V_2) e^{i (x\,XX + y\,YY + z\,ZZ)} (V_3 \otimes V_4), \nonumber\\
    &\quad g \in \{1, i\},\, \pi/4\geq x\geq y\geq |z|
\end{align}
The canonical coordinate $(x,y,z)$ which determines a class of two-qubit gates locally equivalent to the so-defined \emph{canonical gate}
\begin{align}
    \mathrm{Can}(x,y,z) = e^{i\Vec{\eta}\cdot\Vec{\Sigma}} = e^{i (x\,XX + y\,YY + z\,ZZ)}
\end{align}
is confined to a tetrahedron dubbed the Weyl chamber\footnote{Note that for the convenience of proof, herein we select the definition of Canonical gate $(x,y,z)\sim e^{i \Vec{\eta}\cdot\Vec{\Sigma}}$ that is different from the one $ (x,y,z)\sim e^{- i \Vec{\eta}\cdot\Vec{\Sigma}}$ in the our paper submission.}.

\subsubsection{Problem Statement}
Given $H[a,b,c]$, the minimum evolution time $\tau_\text{opt}$~\cite{hammerer2002characterization} to achieve the Weyl chamber coordinate $(x,y,z)$, which indicates a class of two-qubit gates locally equivalent to the canonical gate $\mathrm{Can}(x,y,z)$, by alternately applying $H[a,b,c]$ and single-qubit Hamiltonians $H_0 \otimes H_1$ satisfies $\tau_\text{opt} = \min \{\tau_1,\tau_2\}$, where
\begin{align}
    a\tau_1 &\geq x, \nonumber\\
    (a+b+c)\tau_1 &\geq x+y+z, \nonumber\\
    (a+b-c)\tau_1 &\geq x+y-z,
\end{align}
and
\begin{align}
    a \tau_2 & \geq \frac \pi 2 - x, \nonumber\\
    (a+b+c) \tau_2 & \geq \frac \pi 2 - x + y -z, \nonumber\\
    (a+b-c) \tau_2 & \geq \frac \pi 2 - x + y +z.
\end{align}

For convenience, we extend the Weyl chamber by combining it with its mirroring copy with the plane $x=\pi/4$. This gives us a larger tetrahedron
\begin{align}
W_{ext} := \big\{ & (x,y,z) \in \mathbb{R}^3 \mid 0 \leq x \leq \pi/2, \nonumber\\
& 0 \leq y \leq \min\{x, \pi/2-x\}, -y \leq z \leq y \big\}.
\end{align}

\begin{figure}[tbp]
    \centering
    \begin{tikzpicture}[
        block/.style={draw, rectangle, minimum height=1.5em, minimum width=2.6em, text width=2.4cm, align=center, font=\small},
        line/.style={draw, -latex}
    ]
    \node[block] (a) {Problem Statement,\\ Canonicalization (\Cref{sec:ps})};  
    \node[block,below left=0.6cm and 0.3cm of a] (b) {Block Diagonalization\\(\Cref{sec:nd_canon})};   
    \node[block,below right=0.6cm and 0.3cm of a] (c) {Reparametrization, \\ Block Diagonalization\\ (\Cref{sec:ea_canon})};  
    \node[block,below=0.6cm of b] (d) {genAshN-ND scheme (\Cref{sec:nd_proof})}; 
    \node[block, below=0.6cm of c] (e) {$(\alpha,\beta)\Leftrightarrow T\Leftrightarrow (x, y, z)$ (\Cref{sec:maps})};   
    \node[block, below=0.6cm of e] (f) {genAshN-EA scheme (\Cref{sec:ea_proof})};  
    \draw[line] (a)-- (b) node[midway, left=0.2cm, font=\scriptsize] {genAshN-ND};  
    \draw[line] (a)-- (c) node[midway, right=0.2cm, font=\scriptsize] {genAshN-EA};  
    \draw[line] (b)-- (d);  
    \draw[line] (c)-- (e); 
    \draw[line] (e)-- (f); 

    \end{tikzpicture}  
    \caption{Overview of the proof.}
    \label{fig:overview}
\end{figure}
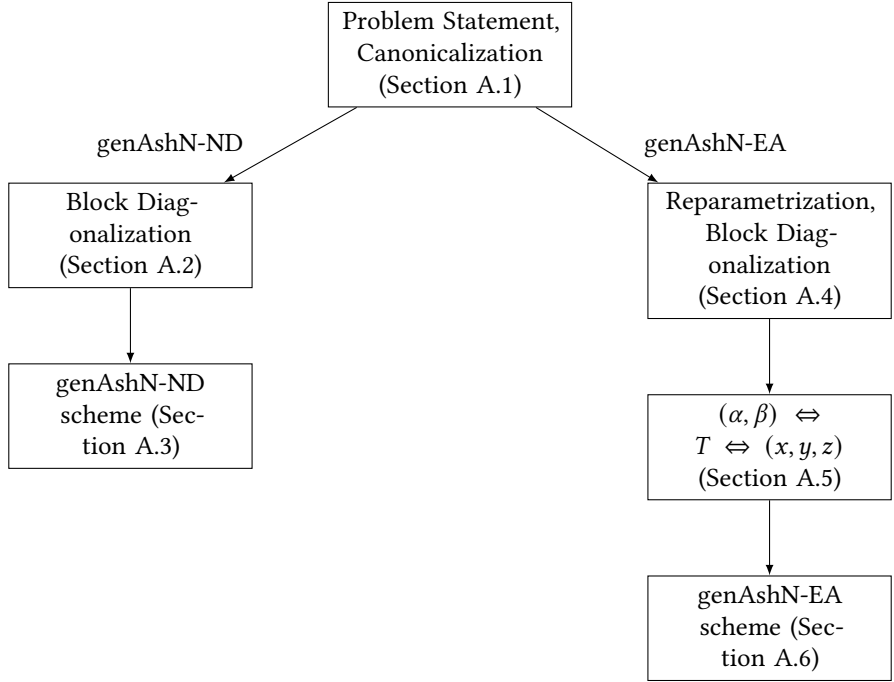

It can be verified that $W_{ext}$ is the union of the collections of the points $(x,y,z)$ and $(\pi/2-x,y,-z)$ where $(x,y,z)$ are canonical Weyl chamber coordinates satisfying $0\leq |z|\leq y\leq x\leq \pi/4$. Note that $(x,y,z)\sim (\pi/2-x,y,-z)$ under local unitaries.

For any canonical Hamiltonian coefficients $(a,b,c)$ and time $\tau\geq 0$, we define the \emph{frontier} $A_\tau$ to be the union of  three polygons, namely 
\begin{align}
    A_\tau := ND \cup EA_+ \cup EA_-,
\end{align}
where we abbreviate $ND:=ND(a,b,c;\tau)$, $EA_\pm:=EA_\pm(a,b,c;\tau)$, and
\begin{align}
    ND &:=\{(x,y,z)\mid x=a\tau,\, y\pm z\in [0, (b\pm c)\tau ]\}, \\
    EA_+ &:=\{(x,y,z)\mid x+y+z = (a+b+c)\tau, \nonumber\\
    &\quad a\tau\geq x\geq y\geq|z|,\, z\geq c\tau \}, \\
    EA_- &:=\{(x,y,z)\mid x+y-z = (a+b-c)\tau, \nonumber\\
    &\quad a\tau\geq x\geq y\geq|z|,\, z\leq c\tau \}.
\end{align}

For a given point $(x,y,z)\in \mathbb{R}^3$, we say that the frontier hits $(x,y,z)$ at time $\tau$ if 
\begin{align}
    \tau = \min\left\{t \mid(x,y,z) \in \bigcup_{t'=0}^t A_t\right\}.
\end{align}
It can be verified that this is equivalent to $\tau$ being the minimum number that simultaneously satisfies
\begin{align}
    \begin{cases}
    a\tau \geq x, \\
    (a+b+c)\tau \geq x+y+z, \\
    (a+b-c)\tau \geq x+y-z,
    \end{cases}
\end{align}
for any coordinate $0\leq |z|\leq y \leq x$. This helps us formulate the optimal gate time as follows:

The optimal gate time to achieve a Weyl chamber coordinate $(x,y,z)$ is equal to the time when the frontier hits either $(x,y,z)$ or $(\pi/2-x,y,-z)$, whichever is earlier.
To prove that a gate scheme achieves optimal time, it then suffices to prove that every (extended) Weyl chamber coordinate $(x,y,z)\in W_{ext}$ can be realized by an appropriately chosen parameter at the time the frontier hits $(x,y,z)$, or equivalently, every point on either of $ND(a,b,c;\tau)\cap W_{ext}, EA_+(a,b,c;\tau)\cap W_{ext}$ and $EA_-(a,b,c;\tau)\cap W_{ext}$ can be realized with $H[a,b,c]$ and local Hamiltonians within time $\tau$.

We prove that this is indeed the case:


\begin{thm}
For each Weyl chamber coordinate $(x,y,z) \in W$ and canonical Hamiltonian coefficients $a \geq b \geq \vert c \vert$, let
\begin{align}
  \tau_\text{opt} &:= \min\{\tau_1, \tau_2\}, \text{ where} \nonumber\\
  \tau_1 &:= \max\left\{ \frac{x}{a}, \frac{x+y+z}{a+b+c}, \frac{x+y-z}{a+b-c} \right\}, \\
  \tau_2 &:= \max\left\{ \frac{\frac{\pi}{2} - x}{a}, \frac{\frac{\pi}{2} - x+y-z}{a+b+c}, \frac{\frac{\pi}{2} - x+y+z}{a+b-c} \right\} \nonumber
\end{align}
be the optimal time as defined above. Then, there exists $\Omega_1, \Omega_2, \delta \in \mathbb{R}$ such that
\begin{align}
    &\exp \{ -i \tau_\text{opt} (H[a,b,c] + (\Omega_1+\Omega_2) X \otimes I \nonumber\\
    &\quad + (\Omega_1 - \Omega_2) I \otimes X + \delta (Z \otimes I + I \otimes Z))\}
\end{align}
has Weyl chamber coordinates $(x,y,z)$. Moreover, at least one of $\Omega_1, \Omega_2, \delta$ is equal to 0.
\end{thm}
\noindent We can prove this via a series of lemmas. 

\begin{lem}[genAshN-ND]
    For any canonical Hamiltonian coefficients $a\geq b\geq |c|$, any $\tau \in(0, \frac\pi{2a}]$, for any Weyl chamber coordinates $(x,y,z)\in ND(a,b,c;\tau)\cap W_{ext}$, there exist $r_1,r_2\geq0$ such that the interaction coefficients for
    \begin{align}
        U:=\exp\{-i\cdot(H[a,b,c] + r_1 \sigma_X\otimes I + r_2 I\otimes \sigma_X)\cdot \tau\}
    \end{align}
    are $(x,y,z)$.
    \label{lem:nd}
\end{lem}

\begin{lem}[genAshN-EA+]
    For any canonical Hamiltonian coefficients $a\geq b\geq |c|$, $\tau \geq0$, for any Weyl chamber coordinates $(x,y,z)\in EA_+(a,b,c;\tau)\cap W_{ext}$, there exist $\Omega,\delta\geq0$ such that the interaction coefficients for 
    \begin{align}
        U := \exp \big\{ -i \cdot \big( & H[a,b,c] + \Omega(\sigma_X\otimes I + I\otimes \sigma_X) \nonumber \\
        & + \delta(\sigma_Z\otimes I + I\otimes \sigma_Z) \big) \cdot \tau \big\}
    \end{align}
    are $(x,y,z)$.
    \label{lem:ea}
\end{lem}


The gate scheme for the genAshN-EA- sector can be derived from genAshN-EA+ according to the following corollary.

\begin{cor}[genAshN-EA-]
    For any canonical Hamiltonian coefficients $a\geq b\geq |c|$, $\tau \geq 0$, for any Weyl chamber coordinates $(x,y,z)\in EA_-(a,b,c;\tau) \cap W_{ext}$, there exist $\Omega,\delta\geq0$ such that the interaction coefficients for 
    \begin{align}
        U:=\exp\big\{-i\cdot \bigl( & H[a,b,c] + \Omega(\sigma_X\otimes I -  I\otimes \sigma_X) \nonumber \\
        &+ \delta(\sigma_Z\otimes I + I \otimes \sigma_Z)\bigr)\cdot \tau \big\}
    \end{align}
    are $(x,y,z)$.
    \label{cor:ea}
\end{cor}

\begin{proof}
    For every $(x,y,z)\in EA_-(a,b,c;\tau)$, we want to prove that there exist $\Omega,\delta>0$ such that the interaction coefficients for
    \begin{align}
        U := \exp \big\{ -i \cdot \big( & H[a,b,c] + \Omega(\sigma_X\otimes I - I\otimes \sigma_X) \nonumber \\
        & + \delta(\sigma_Z\otimes I + I \otimes \sigma_Z) \big) \cdot \tau \big\}
    \end{align}
    are $(x,y,z)$. We consider a similar dual problem of generating the Weyl chamber coordinates $(x,y,-z)$ with single-qubit time-invariant Hamiltonians using $H[a,b,-c]$. It can be verified that $(x,y,-z)\in EA_+(a,b,-c;\tau)$; therefore there exist $\Omega'$ and $\delta'$ such that 
    \begin{align}
        U' := \exp \big\{ -i \cdot \big( & H[a,b,-c] + \Omega'(\sigma_X\otimes I+I\otimes \sigma_X) \nonumber \\
        & + \delta'(\sigma_Z\otimes I + I\otimes \sigma_Z) \big) \cdot \tau \big\}
    \end{align}
    has Weyl chamber coordinates $(x,y,-z)$. As single-qubit gates before and after does not change the Weyl chamber coordinate, the gate
    \begin{align}
        &(\sigma_Z\otimes I)U'(\sigma_Z\otimes I) \nonumber\\
        &=\exp\{-i(H[-a,-b,-c]+\Omega'(-\sigma_X\otimes I+I\otimes \sigma_X) \nonumber\\
        &\quad + \delta'(\sigma_Z\otimes I + I\otimes \sigma_Z)) \tau\} \nonumber\\
        &=\exp\{-i(H[a,b,c]+\Omega'(\sigma_X\otimes I-I\otimes \sigma_X) \nonumber\\
        &\quad - \delta'(\sigma_Z\otimes I + I\otimes \sigma_Z)) \tau\}^{\dag}
    \end{align}
    also has Weyl chamber coordinates $(x,y,-z)$. The gate 
    \begin{align}
    \exp \big\{ -i \cdot \big( & H[a,b,c] + \Omega'(\sigma_X\otimes I - I\otimes \sigma_X) \nonumber \\
    & - \delta'(\sigma_Z\otimes I + I\otimes \sigma_Z) \big) \cdot \tau \big\}
    \end{align}
    then has Weyl chamber coordinates $(x,y,z)$, indicating that $\Omega=\Omega'$ and $\delta=-\delta'$ suffices for the genAshN-EA- scheme.
\end{proof}

\subsubsection{Transforming Weyl chamber coordinates to spectrum}

The rest of the report focuses on proof of~\Cref{lem:nd} and~\Cref{lem:ea}. Before diving into the proofs, we make a few observations that help simplify the problem to be solved.

In the case of genAshN gates, the problem of determining Weyl chamber coordinates can be mapped to determining the spectrum of a locally equivalent gate, or more specifically, gates of the form $\exp\{-iHt\}\cdot (\sigma_Y\otimes \sigma_Y)$.

We first state the intuition for this treatment. For all the genAshN gates, the Hamiltonians involved come from the linear subspace spanned by $\{\sigma_X\otimes \sigma_X, \sigma_Y\otimes \sigma_Y, \sigma_Z\otimes \sigma_Z,\sigma_X\otimes I, I\otimes \sigma_X, \sigma_Z\otimes I + I\otimes\sigma_Z\}.$ It is easy to check that every element in this linear subspace is real and thus symmetric, meaning that all the genAshN gates obtained by exponentiating the genAshN Hamiltonians must be symmetric as well. For such a symmetric unitary $U$, its KAK decomposition should look like $(A_1\otimes A_2) \exp\{i\vec{\eta}\cdot \vec{\Sigma}\}(A_1\otimes A_2)^T$, where $A_1,A_2\in SU(2)$. For a single qubit gate $A\in SU(2)$, one can check that $\sigma_Y A^T \sigma_Y = A^\dag$. Therefore 
\begin{align}
    U\cdot (\sigma_Y\otimes \sigma_Y)&=(A_1\otimes A_2) (\exp\{i\vec{\eta}\cdot \vec{\Sigma}\} \cdot\, \sigma_Y\otimes \sigma_Y)(A_1\otimes A_2)^\dag \nonumber\\
    &\sim \exp\{i\vec{\eta}\cdot \vec{\Sigma}\}\cdot \sigma_Y\otimes \sigma_Y,
\end{align} 
the spectrum of which uniquely determines the Weyl chamber coordinates $\vec{\eta}$.

To make this rigorous, we recall the \emph{magic basis}
\begin{align}
\mathcal{M}:=\frac{1}{\sqrt{2}}\begin{bmatrix}1&0&0&i\\0&i&1&0\\0&i&-1&0\\1&0&0&-i\end{bmatrix}.
\end{align}
The KAK decomposition can be understood under the magic basis as finding $O_1,O_2\in SO(4)$ such that $\mathcal{M}^\dag U \mathcal{M}=O_1DO_2$ where $D$ is a diagonal matrix containing information of the Weyl Chamber coordinates.

Going to the magic basis, for $U=U^T$ we have
\begin{align}
(\mathcal{M}^\dag U \mathcal{M})^T = \mathcal{M}^T U \bar{\mathcal{M}}.
\end{align}
Let 
\begin{align}
D' = \begin{bmatrix}1&&&\\&-1&&\\&&1&\\&&&-1\end{bmatrix}.
\end{align}
One can verify that $\mathcal{M} D'=\bar{\mathcal{M}}$; therefore $\mathcal{M}^\dag U \mathcal{M}D'=\mathcal{M}^\dag U \bar{\mathcal{M}}$ is a symmetric unitary, meaning that it can be diagonalized by an orthogonal matrix:
\begin{align}
\mathcal{M}^\dag U \mathcal{M}D'=O D O^T\sim D.
\end{align}
Translating this back to the computational basis, we have $\mathcal{M}D'=(\sigma_Y\otimes \sigma_Y)\mathcal{M}$. Thus $U \cdot (\sigma_Y\otimes \sigma_Y)\sim D$.

For the rest of the note, we study the two-qubit gate
\begin{align}
    &V(a,b,c;\tau;\Omega_1,\Omega_2, \delta) \nonumber\\
    &:=\exp\{-i\tau(H[a,b,c] + \Omega_1(\sigma_X\otimes I + I\otimes \sigma_X) \nonumber\\
    &\quad + \Omega_2(\sigma_X\otimes I - I\otimes \sigma_X) \nonumber\\
    &\quad + \delta(\sigma_Z\otimes I + I\otimes \sigma_Z))\}\cdot (\sigma_Y\otimes \sigma_Y).
\end{align}

\subsection{Canonicalization for genAshN-ND}
\label{sec:nd_canon}
The Hamiltonian corresponding to the genAshN-ND scheme can be explicitly exponentiated. To see this more clearly, we conjugate the X- and Z-bases:
\begin{align}
    H_{ND}&=(H\otimes H) (H[a,b,c]+\Omega_1(\sigma_X\otimes I +  I\otimes \sigma_X) \nonumber\\
    &\quad +\Omega_2(\sigma_X\otimes I -  I\otimes \sigma_X))(H\otimes H) \nonumber\\
    &=\begin{pmatrix}a + 2\Omega_1 & && c-b\\&-a+2\Omega_2 & b+c&\\&b+c&-a-2\Omega_2&\\c-b&&&a-2\Omega_1 \end{pmatrix}.
\end{align}

This is a $2+2$ block-diagonal matrix which can be explicitly diagonalized. We then have
\begin{align}
    &V(a,b,c;\tau;\Omega_1,\Omega_2,0) \nonumber\\
    &\sim(H\otimes H)V(a,b,c;\tau;\Omega_1,\Omega_2,0)(H\otimes H) \nonumber\\
    &=\exp\{-i H_{ND} \tau\}(\sigma_Y\otimes \sigma_Y).
\end{align}
Define $S_1:=\sqrt{4\Omega_1^2+(b-c)^2}$, $S_2:=\sqrt{4\Omega_2^2+(b+c)^2}$. This can be written in block form (with $V_{ij}$ denoting entries):
\begin{align}
    V_{11} &= \frac{ie^{-i a\tau}(c-b)\sin S_1\tau}{S_1}, \nonumber\\
    V_{14} &= \frac{e^{-ia\tau}(-S_1\cos S_1\tau+i\Omega_1\sin S_1\tau)}{S_1}, \nonumber\\
    V_{22} &= \frac{-ie^{ia\tau}(b+c)\sin S_2\tau}{S_2}, \\
    V_{23} &= \frac{e^{ia\tau}(S_2\cos S_2\tau-i\Omega_2\sin S_2\tau)}{S_2}, \nonumber
\end{align}
with $V_{41}=\bar{V}_{14}$, $V_{44}=V_{11}$, $V_{32}=\bar{V}_{23}$, $V_{33}=V_{22}$. The spectrum of $V$ is $\{ie^{-i(at\pm \theta_1)}, -ie^{i(at\pm \theta_2)}\}$, where $\cos\theta_1 := \frac{(c-b)\sin(S_1t)}{S_1}$ and $\cos\theta_2 := \frac{(b+c)\sin S_2t}{S_2}$. It can be verified that such a gate corresponds to Weyl chamber coordinates $(at, y, z)$, where
\begin{align*}
\sin(y+z)=\frac{(b+c)\sin S_2t}{S_2}, \, \sin(y-z)=\frac{(b-c)\sin S_1t}{S_1}.
\end{align*}

It suffices to prove that the polygon $ND(t)$ can be spanned by choosing appropriate $S_1\geq b-c$ and $S_2\geq b+c$.

\subsection{Proof of \Cref{lem:nd}}
\label{sec:nd_proof}
\begin{proof}
Letting $u = y+z$, $v = y-z$. We claim that $u\in[0, \min\{(b+c)\tau, \pi-(b+c)\tau\}]$. This is because
\begin{align}
    u=y+z\leq 2y\leq \pi - 2x=\pi - 2a\tau\leq \pi - (b+c)\tau.
\end{align}
Similarly, $v\in[0, \min\{(b-c)\tau, \pi-(b-c)\tau\}]$. 
We then only need to prove that the range of $(b-c)\sin(S_1\tau)/S_1$ for $S_1\ge b-c$ covers $[0, \sin((b-c)\tau)]$ and the range of $(b+c)\sin(S_2\tau)/S_2 $ for $S_2\ge b+c$ covers $[0, \sin((b+c)\tau)]$, both of which are obvious.
\end{proof}

\subsection{Canonicalization \& Reparameterization for genAshN-EA}
\label{sec:ea_canon}

For the genAshN-EA scheme, we first rescale the Hamiltonian such that $c = a-1$ and denote $\eta:=a-b\in[0,1]$. The only case where such a canonicalization cannot be done is on the line $a=b=c$, but a closer inspection shows that $EA_+(a,a,a;\tau)=\{(a\tau,a\tau,a\tau)\}$ only consists of one point, in which case \Cref{lem:ea} automatically holds for $\Omega, \delta = 0$. Note also that this rescaling implies $a \geq -c = 1-a$, so $ a \geq \frac 1 2$. 
Furthermore, one can show that with such a rescaling, the highest value $\tau_\text{opt}$ can take over the Weyl chamber is $\pi$:
\begin{proof}
    Suppose $\tau_{opt}=\min\{\tau_1,\tau_2\}>\pi$, i.e. $\tau_1,\tau_2>\pi$. For this to happen it must be that $(a+b+c)\tau_1=x+y+z$ and $(a+b+c)\tau_2=\pi/2-x+y-z$, as other inequalities cannot be saturated  with $a \geq 1/2$ and $(x,y,z)\in W_{ext}$. However this is also impossible as it leads to 
    \begin{align}
        \pi < (a+b+c)(\tau_1+\tau_2)=\pi/2+2y\leq \pi.
    \end{align}
\end{proof}

The Hamiltonian reads
\begin{align}
    H_{EA}&:=a(\sigma_X\otimes \sigma_X + \sigma_Y\otimes \sigma_Y + \sigma_Z\otimes \sigma_Z) \nonumber\\
    &\quad -\eta\cdot \sigma_Y\otimes \sigma_Y-\sigma_Z\otimes \sigma_Z \nonumber\\
    &\quad + \Omega(\sigma_X\otimes I + I\otimes \sigma_X) \nonumber\\
    &\quad + \delta(\sigma_Z\otimes I + I\otimes \sigma_Z).
\end{align}

We start by observing that all terms have a common eigenvector $(0,1,-1,0)$, and it is consequently an eigenvector of $H_{EA}$ with eigenvalue $1+\eta-3a$. The other three eigenvalues are difficult to solve for explicitly, making it difficult to directly exponentiate the Hamiltonian.\footnote{Technically such an explicit exponentiation can be done as the remaining characteristic polynomial is only cubic. We proceed with another approach that will prove to be simpler.} We instead look at the remaining characteristic polynomial:
\begin{align}
    P(\lambda):=\frac{\det(H_{EA}-\lambda I)}{\lambda + 3a-1-\eta}.
\end{align}

$P$ is a cubic polynomial with leading coefficient $+1$; moreover it can be verified that $P(a+1-\eta)=-8\Omega^2\leq 0$, and $P(a-1-\eta)=8\delta^2\geq 0$. This indicates that the three real roots of $P$ lies in $(-\infty,a-1-\eta],[a-1-\eta,a+1-\eta]$ and $[a+1-\eta,+\infty)$ respectively. Moreover, the three roots should sum up to $3a-1-\eta$ as $H_{EA}$ is traceless. We can uniquely reparameterize the eigenvalues by $a+\eta-1-2(\alpha+\beta),a-1-\eta+2\alpha, a+1-\eta+2\beta$, where $\alpha\in[0,1],\beta\geq 0,\alpha+\beta\geq \eta$. We denote the set $Q_\eta:=\{(\alpha,\beta)\in\mathbb{R}^2|\alpha\in[0,1],\beta\geq 0,\alpha+\beta\geq \eta\}$. The last constraint is to ensure that $a+\eta-1-2(\alpha+\beta)\leq a-1-\eta$. Inverting the parameterization by looking at the coefficients of $P(\lambda)$ gives $\Omega= \sqrt{(1-\alpha)\beta(1-\eta+\alpha+\beta)}$ and $\delta= \sqrt{\alpha(1+\beta)(\beta+\alpha-\eta)}$. Since $\Omega, \delta$ can take any non-negative value, this parameterization constitutes a bijection.

\subsection{Correspondence between $(\alpha,\beta)$, $T$ and $(x,y,z)$}
\label{sec:maps}

Given the aforementioned reparameterization, the eigenspectrum of $H_{EA}$ can be determined explicitly. This allows us to derive the trace of $V(a, a-\eta, a-1; \tau; \Omega, 0, \delta)$ as a function of the parameters, where $\Omega = \sqrt{(1-\alpha)\beta(1-\eta+\alpha+\beta)}$ and $\delta = \sqrt{\alpha(1+\beta)(\beta+\alpha-\eta)}$:
\begin{align}
    T(a,\eta;\tau;\alpha,\beta)&:=e^{ia\tau}(\mathrm{Tr}[V]+\exp\{-i\tau(1+\eta-3a)\}).
    \label{eqn:deft}
\end{align}
We claim that showing the range of $T$ covers a certain set $S \subseteq \mathbb{C}^2$ is sufficient to prove~\cref{lem:ea}. To see this, we investigate the correspondence between the eigenvalue parameters $(\alpha,\beta)$, trace quantity $T$ and the Weyl chamber coordinates $(x,y,z)$.
Let us take a step back to revisit what we want to prove: We want to show that the Weyl chamber coordinates of $V(a,a-\eta,a-1;\tau;\Omega,0,\delta)$ maps the pair $(\Omega,\delta)$ surjectively onto $EA_+(a,a-\eta,a-1;\tau)$. To prove this, we want to find a series of three surjective maps that compose to this particular map:
\begin{align*}
(\Omega,\delta)\in\mathbb{R}^2_{\geq 0} \xrightarrow{\phi_1} (\alpha,\beta)\in Q_\eta\xrightarrow{\phi_2} T\in S \xrightarrow{\phi_3} (x,y,z)\in EA_+.
\end{align*}
We have already worked out $\phi_1$, showing that it is a bijection. The rest of the proof focuses on finding $\phi_2,\phi_3$ as well as the set $S\subseteq\mathbb{C}^2$ such that surjectivity holds.

The map $\phi_2$ is directly defined by~\Cref{eqn:deft}, but the map $\phi_3$ is less obvious. To see why $T$ determines the Weyl chamber coordinate $(x,y,z)$ of $V$, it suffices to show that the spectrum of $V$ can be recovered from $T$. Since $V\in SU(4)$, the four eigenvalues have the form $-\exp\{-i\tau(1+\eta-3a)\}, x_1,x_2,x_3$, where
\begin{align}
\begin{cases}x_i^{-1}=\bar{x}_i, i=1,2,3,\\x_1x_2x_3 = -\exp\{i\tau(1+\eta-3a)\},\\ x_1+x_2+x_3 = e^{-ia\tau}T.\end{cases}
\end{align}

This indicates that $x_1,x_2,x_3$ are the three roots of the cubic polynomial
\begin{align}
    &(x-x_1)(x-x_2)(x-x_3) \nonumber\\
    &=x^3 - (x_1+x_2+x_3)x^2 \nonumber\\
    &\quad + (x_1x_2+x_1x_3+x_2x_3)x - x_1x_2x_3 \nonumber\\
    &=x^3 - e^{-ia\tau}T x^2 -\exp\{i\tau(1+\eta-3a)\} \nonumber\\
    &\quad \times ((x_3^{-1}+x_2^{-1}+x_1^{-1})x-1) \nonumber\\
    &=x^3 - e^{-ia\tau}T x^2 -\exp\{i\tau(1+\eta-3a)\} \nonumber\\
    &\quad \times (e^{ia\tau}\bar{T}x-1),
\end{align}
thus establishing $\phi_3$.

We prove the surjectivity of $\phi_2$ and $\phi_3$ using the following respective approaches:

\begin{itemize}
    \item For $\phi_2$, we prove surjectivity using continuity: Given that the map $\phi_2$ is continuous in a sense we will define, we take a contour $C\subseteq Q_\eta$ in the domain and inspect its image $\phi_2(C)$. As the map is continuous, the image of the region enclosed by $C$ must contain the region enclosed by $\phi_2(C)$, which we show contains $S$.
    \item For $\phi_3$, we look at the fibers $\phi_3^{-1}(x,y,z)$. Each $(x,y,z)$ gives rise to a finite number of choices of the spectrum of $V$ due to different possible canonicalizations, which in turn gives rise to a finite number of $T$. We fix one particular canonicalization 
    \begin{align}
    (x,y,z) \mapsto \bigl( & -e^{i(x+y+z)}, e^{i(x-y-z)}, \nonumber \\
    & -e^{i(-x+y-z)}, e^{i(-x-y+z)} \bigr)
    \end{align}
    such that $\phi_3$ can be made invertible. Surjectivity is proven by verifying that $\phi_3^{-1}(EA_+)\subseteq S$.
\end{itemize}

\subsection{Proof of \Cref{lem:ea}}
\label{sec:ea_proof}
\subsubsection{Continuity of $\phi_2$}
Let $\gamma:=1/(1+\beta)$ and define $TG(\alpha,\gamma):=\phi_2(\alpha,\beta)$. We show a stronger notion of continuity for $\phi_2$, that is,
{\small
\begin{align}
TG &= \frac{\gamma  (\alpha  \gamma  (\eta -1)+\alpha -\gamma  \eta ^2+\gamma +\eta -1)}{(\gamma  (\alpha -\eta -1)+2) (\gamma  (2 \alpha -\eta -1)+1)} \nonumber \\
& \quad \times e^{-i \tau  (-2 \alpha -2/\gamma +\eta +1)} \nonumber \\
& \,+ \frac{\gamma  (\alpha  \gamma  (\alpha -\eta -1)+\alpha +\eta -1)}{(\alpha  \gamma -1) (\gamma  (\alpha -\eta -1)+2)} \nonumber \\
& \quad \times e^{i \tau  (\gamma  \eta +\gamma -2)/\gamma} \nonumber \\
& \,+ \frac{(\gamma  (\alpha  (\gamma  (-\eta )+\gamma -1)+\eta +1)-1)}{(\alpha  \gamma -1) (\gamma  (2 \alpha -\eta -1)+1)} \nonumber \\
& \quad \times e^{i \tau  (-2 \alpha +\eta +1)} \label{eq:TG}
\end{align}
}
is continuous in the region $U_\eta$ defined by
\begin{align}
  0 \leq \alpha, \gamma \leq 1  \label{eq:cond_1}\\
  \gamma(\alpha-1)+1\geq \eta\gamma \label{eq:cond_2}.
\end{align}

This shows that $\phi_2$ is not only continuous in the interior of $Q_\eta$, it is continuous at the boundary as well, especially when $\beta\rightarrow\infty$. This allows us to take a contour with $\beta\rightarrow\infty$. 

Now according to the expression for $TG$ in~\cref{eq:TG}, we can only have problems with continuity if one or more of the following equations hold:
\begin{align}
    \gamma(\alpha-\eta-1)+2&=0, \nonumber\\
    \gamma(2\alpha-\eta-1)+1&=0, \\
    \alpha\gamma-1 &= 0. \nonumber
\end{align}
When $\gamma = 0$, none of these equations can be satisfied. When $\gamma > 0$, the equations can be translated to 
\begin{align}
    \alpha + 2/\gamma -1 &= \eta, \nonumber\\
    2\alpha+1/\gamma -1 &= \eta, \\
    \alpha &= 1/\gamma. \nonumber
\end{align}
Since $1/\gamma > 0$, the first equality is never satisfied because of~\cref{eq:cond_2}. Assume $0 < \alpha < 1$. Then, the second equality is never satisfied due to~\cref{eq:cond_2}. Furthermore, the third equality is never satisfied since $1/\gamma \in [1, \infty)$. The only two limits we need to compute are therefore
\begin{align}
    \alpha \to 0,\, \gamma \to 1/(1+\eta) \quad \text{and} \quad \alpha \to 1,\, \gamma \to 1
\end{align}
given $\gamma >0$. We compute the first limit. Let $\epsilon >0$ and $h\in \mathbb{R}$ be small. We can show 
\begin{align}
    &TG \vert_{\alpha = \epsilon, \gamma = 1/(1+\eta) + h} \nonumber\\
    &=  e^{i\tau(1+\eta)} \frac{2 (\eta +1) (\eta  (h-1)+h)}{(\eta +(1+\eta) h (-\eta +\epsilon -1)+\epsilon +1)} \nonumber\\
    &\quad \times \frac{1}{(-\eta +\epsilon  (\eta  h+h+1)-1)} \nonumber\\
    &\quad + e^{-i\tau(1+\eta)} O(1) + O(\epsilon +h).
\end{align}
Next, we compute the second limit. Let $\epsilon, \epsilon' > 0$. We can show
\begin{align}
    &TG \vert_{\alpha = 1-\epsilon, \gamma = 1-\epsilon'} \nonumber\\
    &= -e^{-i\tau(-4+\eta+1)} \frac{(\epsilon '-1) ((\eta -1) (\eta +\epsilon ) \epsilon '-\eta  (\eta +\epsilon -2))}{(-\eta +(\eta +\epsilon ) \epsilon '-\epsilon +2)} \nonumber\\
    &\quad \times \frac{1}{(-\eta +(\eta +2 \epsilon -1) \epsilon '-2 \epsilon +2)} \nonumber\\
    &\quad + e^{-i\tau(1-\eta)} \frac{2 (-\eta +(\eta +\epsilon -1) \epsilon '-\epsilon +2) }{(-\eta +(\eta +\epsilon ) \epsilon '-\epsilon +2)} \nonumber\\
    &\quad \times \frac{(-\eta +(\eta +\epsilon ) \epsilon '-\epsilon +1)}{(-\eta +(\eta +2 \epsilon -1) \epsilon '-2 \epsilon +2)} + O(\epsilon+h).
\end{align}
Since $\eta \in [0,1]$, the denominators for both of these expressions do not vanish as $\epsilon, h, \epsilon' \to 0$. Thus, $TG$ converges to its corresponding value throughout $U_\eta$ and is thus continuous on $U_\eta$.


\subsubsection{Surjectivity of $\phi_2$}
\label{subsec:phi2_surj}
When $\alpha,\gamma$ traverses along the boundary of $U_\eta$, we claim $TG$ encloses $S$, a region of the complex plane enclosed by the arc around the unit circle from $D:= e^{i (\eta-1) \tau }$ to $C := e^{i (\eta+1) \tau }$ by increasing the complex argument by $2\tau$, denoted $\arc{DC}$, and the two line segments $\overline{AC}$ and $\overline{AD}$, where $A:= e^{i\tau(\eta-1)} + e^{i\tau(\eta+1)} - e^{i\tau(1-\eta)}$. This defines the boundary of $S$. For clarity we will 
explicitly define the interior of $S$. We can do this by splitting into cases: 

\begin{enumerate}
    \item If $2\tau = 0$, $A = C = D$ and $S := \{1\}$.
    \item If $0< 2\tau < 2\pi$, we can show that each line segment can only intersect $\arc{DC}$ once: For any point $E$, the line segment $\overline{CE}$ intersects with $\arc{DC}$ only if 
    the angle of $\overline{EC}$ lies strictly within the angle of $\overline{DC}$ and the angle of the tangent line at $C$, that is, $[\eta\tau-\pi/2, (\eta + 1)\tau - \pi/2]$. We have $\overline{CA} = -2i\sin((1-\eta)\tau)$ and the angle is $-\pi/2$, which is not strictly in this range. A similar argument applies to the point $D$, where the range is $[(\eta-1)\tau+\pi/2, \eta\tau+\pi/2]$ and the angle of $\overline{DA}$ is $\tau + \pi/2$. Thus, the intersection of $\triangle ADC$ with the segment of the circle corresponding to $\arc{DC}$ is exactly $\overline{DC}$. We then define $S$ as the region enclosed by $\triangle ADC$ but with $\overline{DC}$ replaced by $\arc{DC}$.
    \label{it:small_angle}
    \item If $2\tau = 2\pi$, $C = D$ and we define $S = \overline{CA} \cup \overline{D^1}$, where $\overline{ D^1} $ is the closed unit disk.
\end{enumerate}

Now, we compute the value of $TG$ at the edges:
\begin{align*}    
    TG\vert_{\alpha=0} &= \frac{e^{A_1}\bigl[\gamma(\eta-1)(1-e^{-2A_1}) + (\gamma\eta+\gamma-2)e^{A_2}\bigr]}{\gamma\eta+\gamma-2}, \\
    TG\vert_{\alpha=1} &= \frac{e^{-A_3}\bigl[\gamma\eta e^{A_4} + (\gamma\eta-2)e^{A_5} - \gamma\eta\bigr]}{\gamma\eta-2}, \\
    TG\vert_{\gamma=1} &= -\frac{e^{i(\eta-1)\tau}\bigl[2\alpha + \eta\bigl(e^{2i\tau(\alpha-\eta+1)} - e^{-2i(\alpha-1)\tau} - 1\bigr)\bigr]}{\eta-2\alpha}, \\
    TG\vert_{\gamma=0} &= e^{i\tau(-2\alpha+\eta+1)}, \\
    TG\vert_{\alpha=-\frac{1}{\gamma}+1+\eta} &= \frac{e^{A_1}\bigl[\gamma(\eta-1)(1-e^{-2A_1}) + (\gamma\eta+\gamma-2)e^{A_2}\bigr]}{\gamma\eta+\gamma-2},
\end{align*}
where $A_1 = \frac{i\tau(\gamma\eta + \gamma - 2)}{\gamma}$, $A_2 = \frac{2i\tau}{\gamma}$, $A_3 = \frac{i\tau[\gamma(\eta-1)-2]}{\gamma}$, $A_4 = \frac{2i\tau(\gamma\eta - 2)}{\gamma}$, and $A_5 = \frac{2i\tau[\gamma(\eta-1)-1]}{\gamma}$.
We consider the contour $(\eta,1) - (0,1/(1+\eta)) - (0, 0) - (1,0) - (1,1) - (\eta,1)$ around $U_\eta$. We compute
\begin{align*}
    TG \vert_{(\alpha,\gamma) = (\eta,1)} &= e^{i\tau(\eta-1)} + e^{i\tau(\eta+1)} - e^{i\tau(1-\eta)} = A, \nonumber\\ 
    B:= TG \vert_{(\alpha,\gamma) = (0,1/(1+\eta))} &= e^{-i\tau(1+\eta)} \frac{1-\eta}{1+\eta} + e^{i\tau(1+\eta)} \frac{2\eta}{1+\eta}, \nonumber\\
    TG \vert_{(\alpha,\gamma) = (0,0)} &= e^{i\tau(\eta+1)}  = C, \\
    TG \vert_{(\alpha,\gamma) = (1,0)} &= e^{i\tau(\eta-1)} =D, \nonumber\\
    E:= TG \vert_{(\alpha,\gamma) = (1,1)} &= e^{i\tau(\eta-1)} \frac{2(1-\eta)}{2-\eta}  +e^{-i\tau(\eta-3)} \frac{\eta}{2-\eta}. \nonumber
\end{align*}
We claim the image of the contour contains the boundary of $S$. In particular, $A,B,C$ and $A,E,D$ each lie on a straight line. First, we expand the expression for $TG$ when $\alpha = 0$:
\begin{align}
    TG \vert_{\alpha =0} &= -e^{-i \tau  (\eta +1 -2/\gamma)} \frac{\gamma(\eta-1)}{\gamma(\eta+1)-2} + e^{i \tau  (\eta +1) } \nonumber\\
    &\quad + e^{i \tau  (\eta +1 -2/\gamma)}  \frac{\gamma(\eta-1)}{\gamma(\eta+1)-2} \nonumber\\
    &= e^{i \tau  (\eta +1) } + 2(\eta-1)\frac{\sin[\tau ( \eta +1 -2/\gamma)]}{\eta+1-2/\gamma}\times i,
\end{align}
where we take $\gamma >0$ to simplify the expression (We already know by continuity that it goes to the desired value when $\gamma \to 0$.). Hence, in this case $TG$ lies on a vertical line and contains $C$ and $B$ when we vary $\gamma \in [0, 1/(1+\eta)]$. Thus by continuity it must contain $\overline{BC}$. Furthermore, $TG \vert_{\alpha= -1/\gamma +1 +\eta}$ has the same expression, and thus it lies on a vertical line and contains $B$ and $A$ when we vary $\gamma \in [1/(1+\eta), 1]$. Thus, it contains $\overline{AB}$. Next, we expand the expression for $\gamma=1$:
\begin{align}
    TG \vert_{\gamma=1} = e^{i\tau(\eta-1)} +\frac{2\eta}{\eta-2\alpha}\sin[\tau(\eta-2\alpha)] \times i e^{i\tau}.
\end{align}
Thus, if we vary $\alpha \in [\eta,1]$, we lie on a line segment that contains $A$ and $E$ and thus by continuity contains the line segment $\overline{EA}$. We next consider $\alpha=1$:
\begin{align}
    TG \vert_{\alpha=1} =   e^{i\tau(\eta-1)} + \frac{2\gamma\eta}{\gamma\eta-2}\sin[\tau(\eta-2/\gamma)]\times i e^{i\tau},
\end{align}
where again we take $\gamma >0$. This is again a straight line segment and contains $E$ and $D$ as we vary $\gamma \in [0,1]$. Hence, it contains the line segment $\overline{DE}$. Furthermore, we see that for both $\gamma=1$ and $\alpha=1$, $TG$ takes the form
\begin{align}
    e^{i\tau(\eta-1)} + r ie^{i\tau},
\end{align}
for some variable $r \in \mathbb{R}$ and thus $\overline{EA}$ and $\overline{DE}$ lie on the same line. Finally, from the expression for $\gamma=0$, $TG$ moves from $D$ to $C$ spanning angle $2\tau$ as we range $\alpha$ from $1$ to $0$. The claim thus follows.

Since $\phi_2$ is continuous, it maps a compact set $U_\eta$ in $\mathbb R^2$ to a compact set in $\mathbb C$. We consider the following cases
\begin{enumerate}
    \item $2\tau =0$. Then $S$ is simply $\{1\}$ and $\gamma = 0$ is sufficient to achieve this point. 
    \item When $0 < 2\tau < 2\pi$, $\phi_2(\partial U_\eta)$ is the boundary of a triangle with one of edges replaced by an arc. Thus, it is clear 
    $S$ is the only compact set with boundary $\phi_2(\partial U_\eta)$. 
    \item When $2\tau = 2\pi$, $\phi_2(\partial U_\eta) = \overline{CA} \cup S^1$, where $S^1$ is the unit circle. It is clear the only compact set with this boundary is $S = \overline{CA} \cup \overline{D^1}$. 
\end{enumerate}
Hence $\phi_2$ is surjective.

\subsubsection{Surjectivity of $\phi_3$}
Now, we know $\tau \in [0,\pi]$. We claim that the image of the Weyl chamber region $W\supseteq EA_+$ defined as the polygon $PQRM$, where
\begin{align}
    P &:=(a\tau, (a-\eta)\tau,(a-1)\tau), \nonumber\\
    Q &:= (a \tau, (a-(1+\eta)/2)\tau, (a-(1+\eta)/2)\tau), \nonumber\\
    R &:= ((a-(1+\eta)/3)\tau, (a-(1+\eta)/3)\tau, \nonumber\\
    &\quad (a-(1+\eta)/3)\tau), \\
    M &:= ((a-\eta/2)\tau, (a-\eta/2)\tau, (a-1)\tau), \nonumber
\end{align}
via $\phi_3^{-1}$ (with respect to our fixed canonicalization) is exactly $S$. Now, $\phi_3^{-1}$ maps a coordinate $(x,y,z)$ in the Weyl chamber to the complex number
\begin{align}
\label{eq:weyl_trace}
    e^{i\tau a} (e^{i(x-y-z)} - e^{i(-x+y-z)} + e^{i(-x-y+z)} ).
\end{align}
Thus,
\begin{align}
   \phi_3^{-1}: P &\mapsto e^{i\tau(\eta+1)} - e^{i\tau(1-\eta)} + e^{i\tau(\eta-1)} = A, \nonumber\\
   Q &\mapsto e^{i\tau(\eta+1)} - 1 + 1 = e^{i\tau(\eta+1)} = C, \nonumber\\
   R &\mapsto e^{i\tau(\eta+1)/3}, \\
   M &\mapsto e^{i\tau} - e^{i\tau} + e^{i\tau(\eta-1)} = e^{i\tau(\eta-1)} = D. \nonumber
\end{align}
We consider the image of $\overline{PQ}$ parameterized by $(a\tau, (a-\eta+t(\eta-1)/2)\tau, (a-1-t(\eta-1)/2)\tau)$, $t\in[0,1]$:
\begin{align}
    & e^{i\tau(\eta-t(\eta-1)/2+1+t(\eta-1)/2)} \nonumber\\
    &\quad - e^{i\tau(-\eta+t(\eta-1)/2 +1+t(\eta-1)/2)} \nonumber\\
    &\quad + e^{i\tau(\eta-t(\eta-1)/2 -1 -t(\eta-1)/2)} \nonumber\\
    &= e^{i\tau(\eta+1)} - e^{i\tau(t-1)(\eta-1)} + e^{-i\tau(t-1)(\eta-1)} \nonumber\\
    &= e^{i\tau(\eta+1)} - 2i \sin[\tau(t-1)(\eta-1)].
\end{align}
This is clearly $\overline{AC}$ since $\tau \in [0, \pi]$ and so the second term cannot change sign. We next consider the image of $\overline{PM}$, parameterized by $((a-t\eta/2)\tau, (a-\eta+t\eta/2)\tau, (a-1)\tau)$, $t \in [0,1]$:
\begin{align}
    & e^{i\tau(-t\eta/2+\eta-t\eta/2+1)} - e^{i\tau(t\eta/2-\eta+t\eta/2+1)} \nonumber\\
    &\quad + e^{i\tau(t\eta/2+\eta-t\eta/2-1)} \nonumber\\
    &= e^{i\tau(\eta+1-t\eta)} - e^{i\tau(1-\eta+t\eta)} + e^{i\tau(\eta-1)} \nonumber\\
    &= e^{i\tau(\eta-1)} + 2\sin[\tau\eta(1-t)] \times i e^{i\tau}.
\end{align}
This is clearly $\overline{AD}$. We next look at the image of $\overline{QR}$, parameterized by $((a-t(1+\eta)/3)\tau, (a-(1+\eta)/2 + t(1+\eta)/6)\tau, (a-(1+\eta)/2 + t(1+\eta)/6)\tau)$, $t\in[0,1]$:
\begin{align}
    &e^{i\tau(-2t(1+\eta)/3 +(1+\eta))} - e^{i\tau(t(1+\eta)/3)} + e^{i\tau(t(1+\eta)/3)} \nonumber\\
    &= e^{i\tau(-2t(1+\eta)/3 +(1+\eta))}.
\end{align}
We also look at the image of $\overline{RM}$, parameterized by $((a-(1+\eta)/3+t(2-\eta)/6)\tau, (a-(1+\eta)/3+t(2-\eta)/6)\tau,(a-(1+\eta)/3-t(2-\eta)/3)\tau)$, $t\in[0,1]$:
\begin{align}
    &e^{i\tau((1+\eta)/3+t(2-\eta)/3)} - e^{i\tau((1+\eta)/3+t(2-\eta)/3)} \nonumber\\
    &\quad + e^{i\tau((1+\eta)/3-2t(2-\eta)/3)} \nonumber\\
    &= e^{i\tau((1+\eta)/3-2t(2-\eta)/3)}.
\end{align}
Hence, we can see explicitly that this traces out $\arc{DC}$. Therefore $\partial W \mapsto \partial S$. Now, $W$ is a compact set in $\mathbb{R}^3$. Since~\cref{eq:weyl_trace} is continuous, the image of $W$ is a compact set in $\mathbb{C}$. The rest of the argument is similar to that of~\cref{subsec:phi2_surj}. The claim thus follows.

\end{document}
\endinput